\documentclass[11pt]{amsart}

\usepackage{amsmath, amssymb, amsthm, amscd}
\usepackage{url}
\usepackage{graphicx}
\usepackage[hidelinks,pagebackref,pdftex]{hyperref}
\usepackage{color}
\usepackage{import}
\usepackage[margin=3cm]{geometry}
\usepackage{caption}
\usepackage{bbold}

\usepackage[algo2e]{algorithm2e} 

\usepackage{array,multirow}

\usepackage{float}
\usepackage[section]{placeins}

\usepackage{subcaption}

\usepackage{hyperref}

\usepackage{marginnote}

\long\def\@savemarbox#1#2{\global\setbox#1\vtop{\hsize\marginparwidth 
  \@parboxrestore\tiny\raggedright #2}}
\renewcommand*{\backref}[1]{}
\renewcommand*{\backrefalt}[4]{
  \ifcase #1
  [No citations.]
  \or [#2]
  \else [#2]
  \fi }

\AtBeginDocument{
   \def\MR#1{}
}

\numberwithin{equation}{section}
\theoremstyle{plain}
\newtheorem{theorem}[equation]{Theorem}
\newtheorem{conjecture}[equation]{Conjecture}
\newtheorem{lemma}[equation]{Lemma}

\newtheorem*{namedtheorem}{\theoremname}
\newcommand{\theoremname}{testing}

\newtheorem{definition}[equation]{Definition}
\newcommand{\R}{{\mathbb{R}}}
\newcommand{\Z}{{\mathbb{Z}}}
\newcommand{\N}{{\mathbb{N}}}
\newcommand{\C}{{\mathbb{C}}}
\newcommand{\Q}{{\mathbb{Q}}}
\newcommand{\tri}{T}
\newcommand{\M}{M}
\newcommand{\TV}{\mathrm{TV}}
\newcommand{\adm}{\mathrm{Adm}}

\newcommand{\SL}{\mathrm{sl}}

\newcommand{\co}{\colon\thinspace}

\newcommand{\ie}{\emph{i.e.}}
\newcommand{\eg}{\emph{e.g.}}
\newcommand{\Pk}{kP}

\newcommand{\PT}{P_{\tri}}
\newcommand{\polydeg}{\mathrm{deg}}
\newcommand{\vol}{\mathrm{vol}}

\begin{document}

\title[Computation of Large Asymptotics of Quantum Invariants]{\Large Computation of Large Asymptotics of 3-Manifold Quantum Invariants}
\author{Cl\'ement Maria \and Owen Rouill\'e}
\thanks{INRIA Sophia Antipolis-M\'editerran\'ee, \url{clement.maria@inria.fr}, \url{owen.rouille@inria.fr} } 

\date{}

\begin{abstract} \small\baselineskip=9pt 
Quantum topological invariants have played an important role in computational topology, and they are at the heart of major modern mathematical conjectures. In this article, we study the experimental problem of computing large $r$ values of Turaev-Viro invariants $\TV_r$. We base our approach on an optimized backtracking algorithm, consisting of enumerating combinatorial data on a triangulation of a 3-manifold. We design an easily computable parameter to estimate the complexity of the enumeration space, based on lattice point counting in polytopes, and show experimentally its accuracy. We apply this parameter to a preprocessing strategy on the triangulation, and combine it with multi-precision arithmetics in order to compute the Turaev-Viro invariants. We finally study the improvements brought by these optimizations compared to state-of-the-art implementations, and verify experimentally Chen and Yang's \emph{volume conjecture} on a census of closed 3-manifolds.
\end{abstract}

\maketitle

\section{Introduction}
\label{sec:introduction}

A main objective of low-dimensional topology is to study and classify 3-manifolds up to homeomorphisms. To do so, topologists have designed a variety of \emph{topological invariants}, \ie, properties allowing to distinguish between non-homeomorphic manifolds. 

A remarkable family of topological invariants are the quantum invariants of Turaev-Viro~\cite{turaev92-invariants}. They consist of an infinite family of real-valued quantities $(\TV_r)_r$, indexed by an integer $r \geq 3$, defined on a triangulation of a manifold. Each invariant $\TV_r$ is made of an exponentially large sum of weights associated to combinatorial data attached to the triangulation.

These invariants have played a key role in computational topology. On the one hand, their combinatorial nature has allowed the design of algorithms and the study of their computational complexity. Defined as a sums over combinatorial states, they can be computed with a simple backtracking algorithm of exponential complexity. They also admit a fixed parameter tractable algorithm, relying on dynamic programming and parameterized by the tree-width of the dual graph of the triangulation~\cite{DBLP:journals/jact/BurtonMS18}. Algorithms can be further improved by pruning the search space for enumeration~\cite{DBLP:conf/esa/MariaS16}. These algorithms are implemented in the major topology software packages \emph{Regina}~\cite{Burton12CompTopWRegina,regina} and the \emph{Manifold Recogniser}~\cite{matveev03-algms,recogniser}, and permit efficient computation of invariants $\TV_r$ for small values of $r$ in practice. These invariants remain however extremely hard to compute in general, and their computation belongs to the complexity class \#P-hard~\cite{DBLP:journals/jact/BurtonMS18,kirby04-nphard}, which most likely bonds any algorithm to rely on the enumeration of exponentially large sets of states. Most notably, computing $\TV_r$ for large values of $r$ is a theoretical and practical challenge, even on small input triangulations.

On the other hand, these topological invariants are remarkably efficient to distinguish between non-equivalent manifolds in practice~\cite{DBLP:conf/esa/MariaS16,matveev03-algms}. As a consequence, their efficient implementations have played a fundamental role in the composition of census databases of 3-manifolds (by work of Burton, Matveev, Martelli, Petronio), which are analogous to the well-known dictionaries of knots \cite{burton07-nor7,matveev03-algms}. These databases are of importance for mathematicians, \eg, to verify conjectures on wide sets of manifolds.

Turaev-Viro invariants are also central in mathematics. They are notably at the heart of Chen and Yang's \emph{volume conjecture}~\cite{ChenY2018,DetcherryKY2018}, the 3-manifold counterpart of the famous \emph{volume conjecture} for knots~\cite{Kashaev97,Murakami01}. Chen and Yang's conjecture states that the asymptotic behavior of the sequence $(\TV_r)_{r \geq 3}$ for a fixed manifold $\M$ is connected to the \emph{simplicial volume} of $\M$, a quantity of distinctively different nature. The volume conjectures in low dimensional topology have attracted a lot of effort from the topology community. They have been partially proved on restricted families of manifolds and knots~\cite{ChenY2018,detcherry2017gromov}, and verified experimentally on few spaces with very specific structure~\cite{ChenY2018}; they remain some of the major open conjectures in low dimensional topology.

\medskip

\paragraph{Contributions.} In this article, we take a different turn on the computation of Turaev-Viro invariants, and compute large $r$ values of $\TV_r$ on large sets of triangulated 3-manifolds. This computation being extremely challenging, we focus our attention to small triangulations of closed 3-manifolds ($n \leq 9$ tetrahedra), and the main dependence in complexity is in $r$, which contrasts with the usual computation where $r$ is small. Under this model, we study experimentally the existing algorithms---backtrack enumeration and fixed parameter tractable algorithm---and conclude that the backtracking algorithm is more appropriate for large $r$ computations of $\TV_r$ in practice (Section~\ref{sec:backtrackvsfpt}). We then resort to optimizing computation with the backtracking method. To do so, we introduce techniques and easily computable parameters to estimate the complexity of the computation of $\TV_r$, for all $r$, on a triangulation $\tri$, using the backtracking algorithm (Section~\ref{sec:ehrhart}). This in turn permits efficient preprocessing of the triangulation---cherry-picking the best triangulation for computation (Section~\ref{sec:expcolor})---, and the fine tuning of arithmetic precision for correct multi-precision computation (Section~\ref{sec:multiprec}). These optimizations are validated by various experimental studies. 
Finally, using these optimizations, we compute Turaev-Viro invariants on large sets of low complexity 3-manifolds, and for large $r$ values. These computations on the first manifolds of low complexity allow us to provide further experimental evidences to the volume conjecture for closed 3-manifolds (graph manifolds and hyperbolic manifolds), and study the speed of convergence of the sequence $(\TV_r)_r$ (Section~\ref{sec:convergence}).

\section{Background}
\label{sec:background}

\paragraph{Manifolds and generalized triangulations.} Unless mentioned otherwise, we consider in this article closed oriented 3-manifolds, \ie, compact oriented manifolds locally homeomorphic to $\R^3$, represented by \emph{generalized triangulations}. A generalized triangulation $\tri$ of a closed $3$-manifold $\M$ is a collection of $n$ abstract tetrahedra $T = \{ \Delta_1, \ldots, \Delta_n \}$ together with $2n$ {\em gluing maps} identifying their $4n$ triangular faces in pairs, such that the underlying topological space is homeomorphic to $\M$. Generalized triangulations are a generalization of simplicial complexes, and they can encode compactly a wide range of $3$-manifolds. For example, there are $13,400$ \emph{topologically distinct}, prime closed oriented 3-manifolds that can be triangulated using at most $11$ tetrahedra, 

We denote by $v$, $e$, $f$, and $n$ the number of vertices, edges, triangles, and tetrahedra of a triangulation $\tri$, after gluings of the tetrahedra. The number of tetrahedra $n$ of $\tri$ is the {\em size} of the triangulation. By standard topological arguments, involving Euler characteristic and Poincar\'e duality, a triangulation of a closed 3-manifold satisfies: 
\begin{equation}
\label{eqn:euler}
e=n+v
\end{equation}

Note that due to the flexibility of the gluings of tetrahedra, we can construct $n$-tetrahedra triangulations with a single vertex, \ie, where all $4n$ vertices of the $\Delta_i$, $ 1 \leq i \leq n$, are identified to one single point. These \emph{$1$-vertex triangulations} are particularly useful for computation. 

We refer the reader to~\cite{jaco03-0-efficiency} for more details on low dimensional topology and generalized triangulations. 

\medskip 

\paragraph{Turaev-Viro invariants.} 
Turaev-Viro type invariants~\cite{turaev92-invariants} of a triangulation $\tri$ are defined as exponentially large sums of weights associated to \emph{admissible colorings} of the edges of $T$.

Let $\tri$ be a generalized triangulation of a closed $3$-manifold $\M$,
let $r \geq 3$ be an integer, and let $I = \{0, 1/2, 1, 3/2, \ldots, (r-2)/2\}$ be the set of the first $r-1$ positive half-integers.

Let $E$ be the set of edges of $\tri$. A \emph{coloring} of $\tri$ is defined to be a map $\theta\co E \to I$ from the edges of $\tri$ to $I$. A coloring $\theta$ is \emph{admissible} if, for each triangle of $\tri$, the three edges $e_1$, $e_2$, and $e_3$ bounding the triangle satisfy the following constraints: 
\begin{equation}\label{eq:cstr1}
	\text{parity condition:} \ \ \theta(e_1)+\theta(e_2)+\theta(e_3) \in \Z;
\end{equation}
\begin{equation}\label{eq:cstr2}
\begin{array}{l}
	\text{triangle inequalities:} \ \ \theta(e_1) \leq \theta(e_2) + \theta(e_3),\\
    \theta(e_2) \leq \theta(e_1) + \theta(e_3), \ \text{and} \   
    \theta(e_3) \leq \theta(e_1) + \theta(e_2); \\
\end{array} 
\end{equation}
\begin{equation}\label{eq:cstr3}
\text{upper bound:} \ \ \theta(e_1)+\theta(e_2)+\theta(e_3)\leq 
    r-2.
\end{equation}
The set of admissible colorings of $\tri$ is denoted by $\adm(\tri,r)$.

Note that in generalized triangulations, a triangle may be defined by less than three different edges, in which case the constraints need to be slightly adapted. This is a simple technicality we do not mention in the following.

Given an admissible coloring $\theta$ of $\tri$, a \emph{weight} $|x|_\theta$ of a face (vertex, edge, triangle, or tetrahedron) of $\tri$ is a complex number that depends exclusively on the colorings $\theta(e)$ of the edges $e$ incident to $x$. 

Given an integer $r$ and a weight system, the Turaev-Viro invariant $\TV_r(\tri)$ associated to a triangulation is defined by:
\begin{equation}\label{eq:tv}
\sum_{\theta \in \adm(\tri,r)} \ \ \prod_{x \text{ face of } \tri} |x|_\theta.
\end{equation}

In the definition of their invariant, Turaev and Viro describe sufficient conditions for the weight system to construct a topological invariant with the above formula, \ie, $\TV(\tri) = \TV(\tri')$ for any two triangulations $\tri$ and $\tri'$ of the same manifold $\M$. They also describe a specific weight system associated to the Lie algebra $\SL_2(\C)$, and depending of an extra parameter $q \in \{1,2\}$, leading to interesting topological invariants. We refer to~\cite{turaev92-invariants} for exact definitions of these ; we use them in our computations for $q=2$, which corresponds to the Chen-Yang volume conjecture verified in Section~\ref{sec:convergence}. 

Note that this article focuses on the efficient enumeration of the set of admissible colorings $\adm(\tri,r)$ of a triangulation, and is consequently mostly generalizable to any weight system. 

\medskip

\begin{table}[t]
\centering
\begin{tabular}{|l|c|c|c|c|c|c|c|c|}
\hline
\#tetra.     & 1 & 2 & 3 & 4 & 5 & 6 & 7 & 8 \\
\hline
\#Mani.      & 3 & 7 & 7 & 14 & 31 & 74 & 175 & 436 \\
\hline
\#$H_1=0$    & 2 & 3 & 3 & 8 & 15 & 33 & 78 & 193 \\
\hline
\end{tabular}
\caption{Number of topologically distinct, prime, closed, orientable 3-manifolds (``\#Mani.'') with a minimal triangulation of size $n$, $1 \leq n \leq 8$, (``\#tetra.'') and those among them with trivial first homology group $H_1(\M,\Z/2\Z)=0$ (``\#$H_1=0$'').}
\label{fig:tabletopologies}
\end{table}

\paragraph{Computation, data, and their representation.} The computation of Turaev-Viro invariants belongs to the complexity class \#P-hard, which makes them extremely challenging to compute for large triangulations or large values $r$. The state-of-the-art algorithms working for any $r \geq 3$ all have worst case complexity of the form $r^{\Theta(n)}$, for $n$ tetrahedra triangulations (more details in Section~\ref{sec:backtrackvsfpt}). 

In this article, we study the computation and values of invariants $\TV_r$ for large values of $r$ (up to $100$). This bonds our computations to \emph{small} triangulations (mostly less than $6$ tetrahedra, and up to $9$ tetrahedra for the most challenging experiments) in order to be feasible. Note however that generalized triangulations can represent a vast variety of distinct topologies with very few tetrahedra. Table~\ref{fig:tabletopologies} summarizes the number of distinct, prime, oriented, closed $3$-manifolds admitting a triangulation of minimal size $n$, for $n \leq 8$. In consequence, the experiments in this article concern tens of topologically distinct topologies, with no topological restriction other than admitting a small triangulation. In comparison, the volume conjecture for closed manifolds, verified experimentally in Section~\ref{sec:convergence} as application of our work, has only been verified experimentally on a few triangulations of closed 3-manifolds~\cite{ChenY2018} that admit the very specific topological property of being constructible by surgeries along two simple knots ($K_{4_1}$ and $K_{5_2}$), which allows for faster computation. The methods in this article improve computation of $\TV_r$ for any triangulated 3-manifold.

\medskip 

In order to present statistics over censuses of triangulations, box and whisker plots are used. For all of them, the line inside the box represents the median of the sample and the box the first and last quartile. The whiskers cover the rest of the sample until 1.5 times the size of the box, every point outside this reach is added to the graph. 

\medskip

All programs are in {\tt C++}, compiled with {\tt gcc 9.3.1}, and run on a Linux machine with 2.40GHz processors and 128GB RAM.

\section{Backtracking vs Parameterized Algorithms}
\label{sec:backtrackvsfpt}

\paragraph{Existing algorithms for Turaev-Viro invariants.} 

In this section, we compare the performance of existing methods to compute $\TV_r$ for large $r$. We start by reviewing algorithms for Turaev-Viro invariants. 

\medskip

\begin{table}[ht]
\addtolength{\tabcolsep}{-1pt}
\centering  
\begin{tabular}{|l|c|c|c|c|c|c|c|c|}
\hline
\#tetra.     & 1 & 2 & 3 & 4 & 5 & 6 & 7 & 8 \\
\hline
\#Mani.      & 3 & 7 & 7 & 14 & 31 & 74 & 175 & 436 \\
\hline
\%faster & 100 & 100 & 100 & 79 & 52 & 34 & 19 & 9 \\
\hline
\end{tabular}
\addtolength{\tabcolsep}{1pt}  
\caption{For fix number of tetrahedra $n \leq 8$ (``\#tetra.''), \% of triangulations in the census for which the backtracking algorithm is faster than the FPT algorithm, when computing $\TV_r$ at $r=11$ with single precision arithmetic.}
\label{fig:tablefaster}
\end{table}

Let $\tri$ be a triangulation of a 3-manifold, and $r \geq 3$ be an integer. There is a straightforward backtracking algorithm to enumerate all admissible colorings $\adm(\tri,r)$ of order $r$, and compute $\TV_r(\tri)$. Let $v$ be the number of vertices and $n$ the number of tetrahedra in $\tri$. Consequently, $\tri$ admits $n+v$ edges (Eq.(~\ref{eqn:euler})). Order the edges $e_1, \ldots, e_{n+v}$ of $\tri$. There are $(r-1)$ possible colors for an edge, consequently the backtracking algorithm consists of traversing a regular $(r-1)$-tree of depth $n+v+1$, where nodes at depth $i$, $1 \leq i \leq n+v$, correspond to the edge $e_i$, and any of the $r-1$ downward facing arcs at such node corresponds to the assignment of a color $\{0, 1/2, \ldots , (r-2)/2 \}$ to $e_i$ in some coloring. In consequence, leaves of the tree are in bijection with colorings. The algorithm consists of traversing this tree in a depth-first fashion, stopping when any of the admissibility constraint (\ref{eq:cstr1}), (\ref{eq:cstr2}), or~(\ref{eq:cstr3}) is violated. If a leaf is reached, it corresponds to an admissible coloring $\theta$, and we keep track of the sum $\sum_{\theta \in \adm(\tri,r)} |\tri|_\theta$. Ordering the edges appropriately~\cite{DBLP:journals/jact/BurtonMS18}, this algorithm runs in worst case $O((r-1)^{n+v})$ operations (machine and arithmetic), storing a constant amount of data (pointers and weights\footnote{Note that depending on the arithmetic precision required, weights may be stored on large numbers of bits.}). However, in practice, only a small portion of the tree is explored, as admissibility constraints are violated early on ; see Section~\ref{sec:expcolor} for an experimental analysis. 

The fixed-parameter tractable (FPT) algorithm introduced in~\cite{DBLP:journals/jact/BurtonMS18} uses a standard dynamic approach to compile the sum (\ref{eq:tv}) defining $\TV_r(\tri)$, run on a combinatorial decomposition of the triangulation $\tri$. When $k$ is a parameter measuring the \emph{sparsity} of the triangulation\footnote{More precisely, $k$ is the \emph{treewidth} of the dual graph of $\tri$~\cite{DBLP:journals/jact/BurtonMS18}.}, with $1 \leq k \leq n-1$, it runs in $O(n \cdot (r-1)^{6(k+1)} k^2 \log(r) )$ operations (machine and arithmetic) when $6(k+1) \leq n+v$, or $O(n \cdot (r-1)^{n+v} (n+v)^2 \log(r) )$ operations otherwise. However, this algorithm has worst case exponential memory, and requires to store $O((r-1)^{6(k+1)})$ arithmetic values. 

As predicted by the complexity analysis of these algorithms, the FPT algorithm shows better performance than the backtracking algorithm on most input triangulations with $n \leq 20$ tetrahedra, for very small values of $r$ ($r \leq 7$)~\cite{DBLP:journals/jact/BurtonMS18}. We study in the following the case of interest for this article, i.e., smaller triangulations $n \leq 9$ and larger values of $r$, and draw different conclusions.

\medskip

\begin{figure}[ht]
\includegraphics[width=0.7\columnwidth]{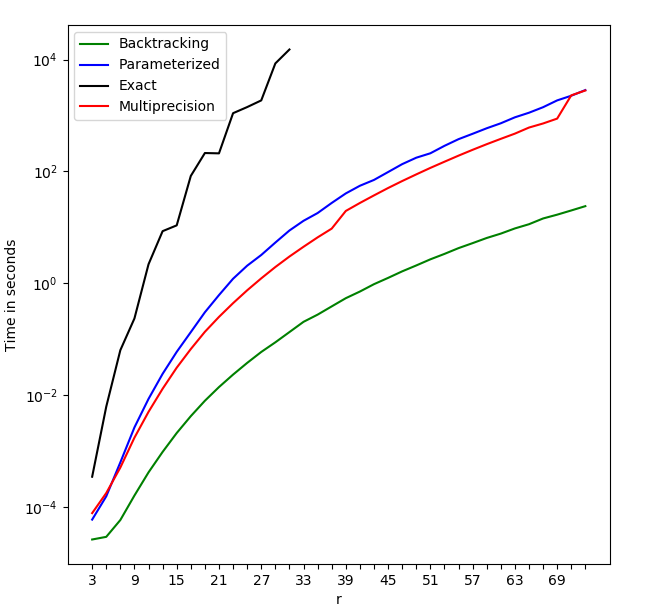}

\caption{Time performance for computing $TV_r$, $3 \leq r \leq 73$, on a 5-tetrahedra, treewidth $4$ minimal triangulation of the manifold $\mathtt{SFS[S^2:(2,1)(3,1)(5,-4)]}$, using the backtracking algorithm with single precision arithmetics (``Backtracking''), with multi-precision arithmetics (``Multiprecision''), and with exact arithmetic in a cyclotomic field (``Exact''), and the FPT algorithm with single precision arithmetic (``Parameterized'').}
\label{fig:pathos}

\end{figure}

\paragraph{Performance of algorithms for large $r$.}

Figure~\ref{fig:pathos} presents the performance of the backtracking and parameterized algorithms on a 5-tetrahedra, treewidth $4$ triangulation of a 3-manifold, with different arithmetic precision. In this analysis, we focus on the single precision backtracking (``Backtracking'') and parameterized (``Parameterized'') algorithms. The computation are run with the {\tt C++} library {\tt Regina}~\cite{regina}. The backtracking algorithm clearly performs orders of magnitude better than the parameterized algorithm, and the difference of performance to compute $\TV_r$ increases with larger values of $r$.

In fact, even if the parameterized algorithm performs better than the backtracking algorithm in \emph{larger} triangulations, as observed in~\cite{DBLP:journals/jact/BurtonMS18}, it is not the case for small triangulation. Table~\ref{fig:tablefaster} pictures the fraction of triangulations of the census for which the backtracking algorithm performs better for the computation of $\TV_r$, with $r=11$. For triangulations with at most $5$ tetrahedra, the backtracking algorithm performs better in the majority of cases, and still remains a better option in a substantial fraction of triangulations with up to $8$ tetrahedra. 

This interpretation is twofold. First, the exponent $6(k+1)$ in the complexity of the parameterized algorithm, $k \geq 1$, is likely to be larger than the exponent $(n+v)$ in the complexity of the backtracking algorithm, when $n$ is small. Second, the backtracking algorithm explores a very small portion of the search tree when enumerating colorings.

Additionally, a main difficulty of the parameterized algorithm is its memory consumption. Indeed, Figure~\ref{fig:pathos} stops at $r=73$ because it is the last order $r$ before which the parameterized algorithm requires more than the $128$GB of memory available in the machine, with single precision numbers (\ie, not taking into account the substantial space required by high precision arithmetic, which is necessary to get a correct result, due to the nature of the computation). 
Recall that the FPT algorithm has exponential memory complexity, which makes it rapidly intractable for larger triangulations and larger values of $r$. In comparison, the backtrack search uses a space linear in the size of the input times the size of the numbers used in the computations.

In conclusion, the backtracking algorithm is a better choice for computing large $r$ asymptotics of the Turaev-Viro invariants on small triangulations. In the following sections, we describe optimizations to and implementation of the backtracking algorithm.

\section{Ehrhart theory for counting admissible colorings}
\label{sec:ehrhart}
In this section, we introduce Ehrhart theory and use it to estimate the number of valid colorings for a subset of the manifold census.

\subsection{Ehrhart theory and counting lattice points in polytopes}

We introduce Ehrhart theory~\cite{EHRH62} on counting lattice points in polytopes.

\begin{definition}
A \emph{polytope} $P$ is a bounded subset of $\R^d$ defined as the set of solutions of a system of linear inequalities:
\[
  P = \{ x | Ax\leq b \}, \text{with} \ \ A\in \R^{m \times d},\ b\in \R^m.
\]
The polytope $P$ is \emph{rational} if $A \in \Q^{m \times d}$ and $b \in \Q^m$.

The \emph{dilatation} of a polytope $P$ by a factor $k\in \N$, denoted by $\Pk$, is the polytope obtained by multiplying $b$ by $k$:

\[
  \Pk = \{ x | Ax\leq kb \}.
\]
It is equivalent to multiplying the coordinates of the vertices of $P$ by $k$.

\end{definition}

Ehrhart theorem states that the number of lattice points, \ie, points with integer coordinates, in the dilatation $\Pk$ of a rational polytope $P$, admits a closed form as a quasi-polynomial. Specifically,

\begin{theorem}[Ehrhart~\cite{EHRH62}]
\label{thm:ehr}
Let $P$ be a rational polytope of dimension $d$. There exists a finite family of periodic functions $\left(e_i \co \N \rightarrow \R \right)_{i=0 \ldots \polydeg}$ such that for all $k\in \N$:
\[
  \#(\Pk\cap\Z^d) = \sum_{i=0}^{\polydeg} k^i \cdot e_i(k).
\]
Where $\polydeg \in \N$. 

Let us denote $L(P,k) = \sum_{i=0\ldots \polydeg} k^i \cdot e_i(k)$. We refer to this function as {\em Ehrhart's polynomial}. It admits several remarkable properties:
\begin{itemize}
\item the period of the coefficients $e_i(\cdot)$ is the smallest integer $k$ such that $\Pk$ is a polytope with integer coordinates vertices;
\item the degree $\polydeg$ of $L(P,\cdot)$ is equal to the dimension $d$ of the polytope $P$;
\item the leading coefficient $e_d(\cdot)$ of $L(P,\cdot)$ is a constant function, whose value is equal to the d-dimensional volume of the polytope $P$.
\end{itemize}
 
\end{theorem}

The quasi-polynomial $L(P, \cdot)$ of a rational polytope $P$ can be computed in time exponential in the embedding dimension:

\begin{theorem}[Barvinok~\cite{Barvinok94,Barvinok05}]
Let $\Delta$ be a $d$-dimensional simplex with integer coordinates vertices. Then $L(\Delta,\cdot)$ can be computed in time $(d\cdot|\Delta|^{O(d)})$ where $|\Delta|$ is the number of bits needed to encode $\Delta$. This result can be extended to rational polytopes.
\end{theorem}

\subsection{Application to Turaev-Viro invariants}
\label{subsec:appTV}

In this section, we adapt the definition of admissible colorings for Turaev-Viro invariants to Ehrhart theory's settings, and describe criteria for triangulations of a manifold $\M$ to admit a minimal number of colorings.

It was observed~\cite{DBLP:conf/esa/MariaS16} that for a 3-manifold $\M$ with trivial homology group $H_1(\M,\Z/2\Z)$, all admissible colorings $\adm(\tri,r)$ of any $1$-vertex triangulation $\tri$ of $\M$ and of any order $r \geq 3$ can only admit integer colors in $\{0, 1, 2, \ldots, \lfloor \frac{r-2}{2} \rfloor \}$. As a consequence, for such triangulations, the backtracking algorithm needs only explore a tree of much smaller size $O( \lfloor \frac{r}{2} \rfloor^{n+v} )$. Additionally, the parity constraints~(\ref{eq:cstr1}) for admissibility, which is the only non-linear constraint, is always satisfied with integer colors and can be omitted. In the following analysis, we exploit this fact to estimate the number of colorings in a triangulation, without enumerating them all, using Ehrhart theory.

As a consequence, we focus our attention to 3-manifolds $\M$ with trivial homology group $H_1(\M,\Z/2\Z)$ for the following theoretical analysis. Note that among the first smallest minimal triangulations of prime closed 3-manifolds, nearly half have trivial $\Z/2\Z$ homology; see Table~\ref{fig:tabletopologies}. We also assume triangulations to be $1$-vertex ; this is not a strong assumption as triangulations can be preprocessed efficiently to reduce the number of vertices to $1$ (see Section~\ref{sec:expcolor}).

\medskip

\paragraph{Analysis and experiments.} The following theoretical analysis follows the case $H_1 = 0$ and $1$-vertex triangulation, to compute an accurate estimator of the number of admissible colorings. In the more general cases (manifolds with non-trivial homology, more than $1$ vertex), the following analysis still provides a {\em lower bound} on the number of admissible colorings of a triangulation, and is still of interest. In consequence, in {\em all} following experiments---except Figure~\ref{fig:worstacc} that concerns the {\em exact} approximation of the number of admissible colorings and Figure~\ref{fig:estimbacktrack} that concerns performances of the backtracking algorithm under this hypothesis---the full census of manifolds is used, regardless of their homology. 

\medskip

Consider now the admissibility constraints~(\ref{eq:cstr2}) and~(\ref{eq:cstr3}), \ie, the constraints defining the admissible colorings minus the parity constraint. They are linear and consequently define a polytope. More specifically, let us consider the following: 

\begin{definition}
\label{def:canon}
Let $\tri$ be a triangulation, with $v$ vertices, $f$ triangles, $n$ tetrahedra, and $n+v$ edges. Consider the following $2(n+v)+4f$ inequalities:

\begin{itemize}
\item for every edge, $-\theta(e) \leq 0$, and $\theta(e) \leq 1/2$,
\item for every triangle with edges $e_1, e_2, e_3$, $\theta(e_1)-\theta(e_2)-\theta(e_3) \leq 0$, $\theta(e_2)-\theta(e_1)-\theta(e_3) \leq 0$, $\theta(e_3)-\theta(e_1)-\theta(e_2) \leq 0$, and $\theta(e_1)+\theta(e_2)+\theta(e_3) \leq 1$.
\end{itemize}

Let $A$ be the $(2(n+v)+4f, n+v)$-matrix with $\{-1,0,1\}$ coefficients, and $b$ the $(n+v)$-vector with $\{0, 1/2, 1\}$ coefficients, defining the rational polytope $\PT = \{x | Ax \leq b\}$ corresponding to this set of equations. We call $\PT$ the \emph{admissibility polytope} associated to $\tri$.
\end{definition}

\begin{lemma}
Let $r\geq 3$ be an integer, $\tri$ a $1$-vertex triangulation, with $n+1$ edges, of a manifold $\M$ with trivial $H_1(\M,\Z/2\Z)$, and $\PT$ the associated admissibility polytope. Then the admissible colorings of $\adm(\tri,r)$ are in bijection with the lattice points $((r-2)\PT \cap \Z^{n+v})$.
\end{lemma}
\begin{proof}
Let $x$ be an integer coordinate vector such that $x \in (r-2)\PT$, \ie, $Ax \leq b$ as defined in Definition~\ref{def:canon}. Interpreting the $i^{\text{th}}$ coordinate $x_i$ of $x$ as an integer color $\theta(e_i)$ for the $i^{\text{th}}$ edge $e_i$ of $\tri$, we note that, by definition of $(r-2)\PT$, all colors $\theta(e_i)$ are positive integers smaller than $(r-2)/2$, and the coloration $\theta$ satisfies admissibility constraints~(\ref{eq:cstr1}),~(\ref{eq:cstr2}), and~(\ref{eq:cstr3}). It consequently defines an admissible coloring.
Because $H_1(\M,\Z/2\Z)$ is trivial and $\tri$ is $1$-vertex, all admissibility colorings appear as such integral vector (see discussion in Section~\ref{subsec:appTV}). \hfill
\end{proof}

Consequently, by Theorem~\ref{thm:ehr}, the number of admissible colorings in $\adm(\tri,r)$ is equal to $L(\PT,r-2)$. We call the quasi-polynomial $L(\PT,\cdot)$ the {\em Ehrhart polynomial} of the triangulation $\tri$.

\section{Triangulations with fewer colorings and estimation of running times}
\label{sec:expcolor}

Choosing the triangulation that admits the smallest number of colorings among a given set of triangulations boils down to computing the Ehrhart polynomials of the triangulations and choosing the one with the minimal value at the desired dilatation. Computing these polynomials is however expensive, and we only consider their asymptotic values, \ie, their degrees and leading coefficients.

\medskip

\subsection{Degree of Ehrhart polynomial.} For a triangulation $\tri$, the degree of $L(P_\tri, \cdot)$ is equal to the dimension of the polytope $\PT$ (Theorem~\ref{thm:ehr}).

\begin{lemma}
\label{lem:dimpoly}
Let $\tri$ be a triangulation and $\PT$ the associated polytope. Then the dimension of $\PT$ equals the number of edges in $\tri$. 
\end{lemma}
\begin{proof}
Vector $(\frac{1}{4},\frac{1}{4},\dots,\frac{1}{4})$ does not saturate any of the constraint defining $\PT$, and consequently lies in the interior of $\PT$. Hence $\PT$ is full dimensional. The dimension of the embedding space being the number of edges, this is also the dimension of $\PT$. \hfill
\end{proof}

As a consequence of Lemma~\ref{lem:dimpoly}, finding a triangulation with smallest asymptotic number of admissible colorings starts with finding a triangulation with smallest number of edges. In a closed triangulation, the number of edges is equal to $n+v$, the number of tetrahedra plus the number of vertices (Equation~\ref{eqn:euler}), which can be controlled:

\begin{theorem}[Jaco and Rubinstein~\cite{jaco03-0-efficiency}]
\label{thm:jrvertices}
Aside from a small set of exceptions, a triangulation $\tri$ of a prime, closed, orientable 3-manifold $\M$, with $n$ tetrahedra and $v$ vertices, can be turned in polynomial time into a $1$-vertex triangulation $\tri'$ of $\M$, with $n' \leq n$ tetrahedra.

Notably, if $\M$ can be triangulated with a minimal number $n_0$ of tetrahedra, $\M$ admit such minimal triangulation with $1$-vertex.
\end{theorem} 

As a consequence, for a fixed prime, closed, orientable 3-manifold $\M$, the number of admissible colorings $\adm(\tri,r)$ for any triangulation $\tri$ of $\M$ grows at least as fast as $\Omega(\left(\frac{r}{c}\right)^{n_0+1})$, where $n_0$ is the size of a minimal triangulation for $\M$ and $c$ is a constant. This is asymptotically optimal. We study the constant $c$, in connection with Ehrhart theory, in the next section.

\subsection{Leading coefficient of Ehrhart polynomial.} The leading coefficient of $L(\PT, \cdot)$ is constant and equal to the volume of the polytope $\PT$ (Theorem~\ref{thm:ehr}). 
Computing the volume of a polytope is $\#$P-hard. However, there exist approximation schemes that are polynomial in the dimension of the polytope. For instance,~\cite{Emiris14} gives an $\epsilon$-approximation of the volume of a $d$-dimensional polytope defined by $m$ hyperplanes in time $O(\epsilon^{-2}md^3\log(d))$ (other terms depending on the geometry of the polytope are hidden). 

We study experimentally the values of the volumes of polytopes associated to minimal triangulations of closed, prime, oriented 3-manifolds from the census.

\begin{figure}[h]
\includegraphics[width=0.7\columnwidth]{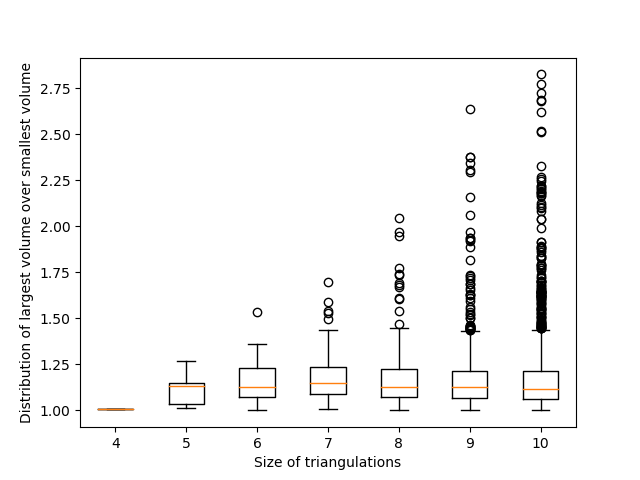}
\caption{Distribution of the ratio of the largest volume over the smaller volume (Equation~\ref{eq:ratiovolminT}) of polytopes associated to minimal triangulations, for all prime closed oriented 3-manifold admitting minimal triangulations of size $4$ to $10$.}
\label{fig:volume_amp}
\end{figure}

\medskip

\paragraph{Volumes of minimal triangulations of a same manifold.} 
Figure~\ref{fig:volume_amp} presents the distribution, over all prime closed orientable 3-manifolds $\M$ admitting minimal triangulations of up to 10 tetrahedra, of the ratio of the largest volume over the smallest volume of any polytope $\PT$ associated to a minimal triangulation of $\M$, \ie, 
\begin{equation}
\label{eq:ratiovolminT}
\displaystyle\frac{\max \{\vol \PT | \tri \ \text{minimal triangulation of} \ \M\} }{\min \{\vol \PT | \tri \ \text{minimal triangulation of} \ \M\} }
\end{equation}
Volumes for triangulations of up to $8$ tetrahedra are exact, computed with Normaliz~\cite{Normaliz}, and volumes for $9$- and $10$-tetrahedra triangulations are approximated within 5\% with the approximation library VolEsti~\cite{Emiris14,VolEsti}.

The ratio of a majority of manifolds is close to $1$. However, more outliers appear with larger triangulations, which is explained by the fact that manifolds tend to have an increasing number of distinct minimal triangulations for a growing number of tetrahedra. Consequently, on certain manifolds of up to 10 tetrahedra, choosing the appropriate minimal triangulation leads to a gain of a factor close to $3$ when minimizing the asymptotic number of admissible colorings. 
Note that manifolds with less than 4 tetrahedra are omitted because they are few with several minimal triangulations. 

In the following, we call a minimal $1$-vertex triangulation $\tri$, with smallest volume $\vol \PT$ over all such triangulations of a manifold $\M$, the \emph{optimal triangulation} for $\M$. In the following experiments, we use optimal triangulations of manifolds.

\medskip

\paragraph{Polytope volumes of increasingly large minimal triangulations.} 
Figure~\ref{fig:volume_distrib} presents the distribution of the volumes of polytopes $\PT$ for optimal triangulations of manifolds in the census, ranging from $1$ to $10$ tetrahedra. We observe a clear exponential decay of the volumes of the polytopes for optimal triangulations of larger sizes. 

By interpolating the medians, and using Ehrhart theory, we deduce that in average, an optimal triangulation $\tri$, with $n+1$ edges, admits experimentally $\Theta(\left[\frac{r-2}{3.97}\right]^{n+1})$ admissible colorings in $\adm(\tri,r)$. This is a (practical) improvement by an exponential factor over the expected $O(\left[\frac{r}{2}\right]^{n+1})$ admissible colorings predicted by the theory for 3-manifolds with trivial homology, and $O(\left[r-1\right]^{n+1})$ in the general case. 

\begin{figure}
\includegraphics[width=0.7\columnwidth]{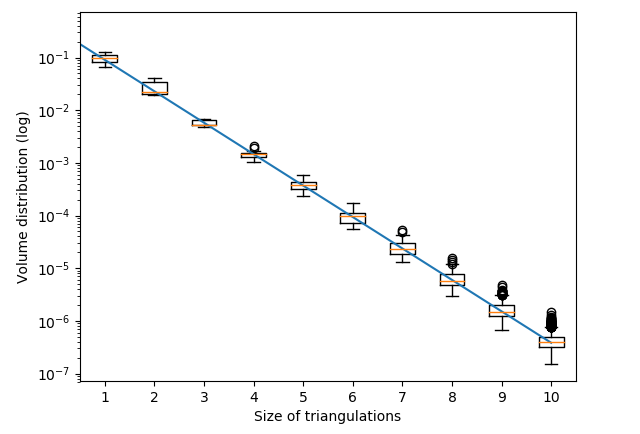}
\caption{Distribution of the volumes of optimal triangulations for all manifolds of the census, up to $10$ tetrahedra. Medians are interpolated (r-value $= -0.9999$) by a straight line of slope $-0.598$ in log scale.}
\label{fig:volume_distrib}
\end{figure}

\subsection{Accuracy of the parameter}

For an optimal triangulation $\PT$ with $n+1$ edges, we study the accuracy of the parameter $\vol \PT \cdot (r-2)^{n+1}$ (\ie, the leading term of the Ehrhart polynomial counting $\adm(\tri,r)$) to estimate the number of admissible colorings of $\tri$, and the running time of the backtracking algorithm.

\paragraph{Estimating the number of admissible colorings.} Figure~\ref{fig:worstacc} presents the ratio:
\begin{equation}
\label{eq:ratioerr}
	\max \left\{\frac{\vol \PT \cdot (r-2)^{n+1}}{ \#\adm(\tri,r) }, \frac{\#\adm(\tri,r)}{\vol \PT \cdot (r-2)^{n+1}}  \right\} \geq 1, 
\end{equation}
for optimal triangulations with $n$ tetrahedra, $1 \leq n \leq 6$, of the manifolds in the census with trivial homology $H_1$. Curves give the worst values of the ratio~(\ref{eq:ratioerr}) over optimal triangulations of the same size.
As we are considering a first order approximation of the Ehrhart polynomials, we observe an expected error of order $O(r^{-1})$, and the quantity $\vol \PT \cdot (r-2)^{n+1}$ appears to be an accurate estimator of the size of $\adm(\tri,r)$ for moderate values of $r$ and above.

\begin{figure}
\includegraphics[width=0.7\columnwidth]{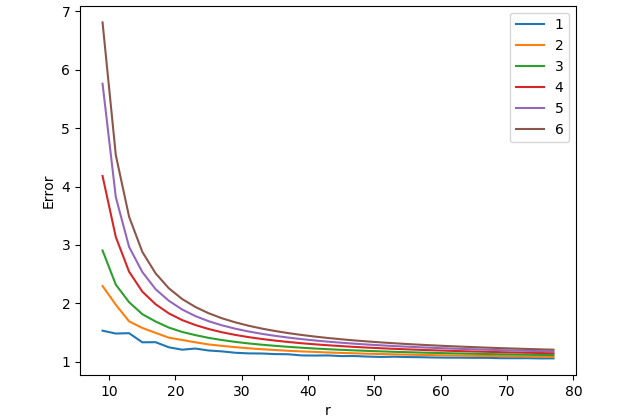}
\caption{Maximal value of ratio (Equation~(\ref{eq:ratioerr}) between the estimator and number of admissible colorings for optimal triangulations with trivial $H_1$ of the census, from $1$ to $6$ tetrahedra, for increasing $r$.}
\label{fig:worstacc}
\end{figure}

\paragraph{Estimating the combinatorial complexity of the backtracking algorithm.} Figure~\ref{fig:estimbacktrack} presents the ratio between the size of the search tree traversed by the backtracking algorithm, and the number of admissible colorings:
\begin{equation}
\label{eq:rationumcolbacktrack}
\frac{\text{size backtracking search}}{\#\adm(\tri,r)} \geq 1,
\end{equation}
for optimal triangulations $\tri$ of the census of $5$-tetrahedra triangulations with trivial homology $H_1$. The size of the backtracking search is an estimate of the number of combinatorial operations the backtracking algorithm for $\TV_r(\tri)$ performs. 
For moderate values of $r$ and above, we observe that the backtracking search is output-sensitive in the size of $\adm(\tri,r)$, and there is no triangulation of size $5$ that traverses a search tree larger that $4 \cdot \adm(\tri,r)$ for $r \geq 39$.

\begin{figure}[h]
\includegraphics[width=0.7\columnwidth]{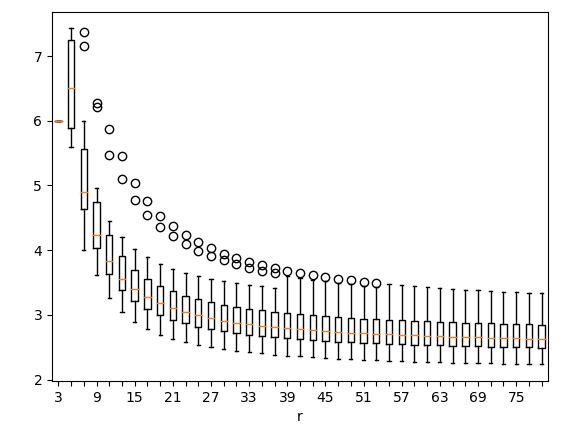}
\caption{Distribution, against $r$, of the number of tree nodes explored by the backtracking divided by the number of valid colorings for all manifolds with trivial homology $H_1$ and smallest triangulation of size $5$.}
\label{fig:estimbacktrack}
\end{figure}

\medskip

In conclusion to this section, given a ($1$-vertex) triangulation $\tri$ of a closed $3$-manifold, we can compute efficiently the leading term $\vol \PT \cdot (r-2)^{n+1}$ of the Ehrhart polynomial, and estimate accurately both the number of admissible colorings in $\adm(\tri,r)$ and the expected running time of the backtracking algorithm to compute Turaev-Viro invariants on $\tri$. We deduce a preprocessing strategy. 

It is worthwhile, when computing large $r$ asymptotics of Turaev-Viro invariants of a manifold $\M$, to find an optimal (or close-to-optimal) triangulation of $\M$ minimizing the leading of the corresponding Ehrhart polynomial. When preprocessing an input triangulation, we consequently reduce it (in polynomial time) into a $1$-vertex triangulation (Theorem~\ref{thm:jrvertices}), and search for a minimal such one. The second step is a theoretically hard procedure, however proceeding to random combinatorial modifications of the triangulation, preserving its underlying topology, such as {\em bistellar flips}, is efficient in practice for reducing the size of triangulations~\cite{Burton12CompTopWRegina}. Additionally, keeping track of the minimal triangulations encountered, we record the volume of their associated polytopes $\PT$, and keep the minimal triangulation of smallest polytope volume.

In the following, we use census triangulations that are already minimal, and we select the optimal such triangulation.

\section{Multi-precision arithmetics}
\label{sec:multiprec}

The enumeration of the colorings and the complexity of the computation of their weights imply a large number of arithmetic operations during the computation of Turaev-Viro invariants. In this section we discuss this aspect of the computation.

\begin{figure}[ht]
\includegraphics[width=0.7\columnwidth]{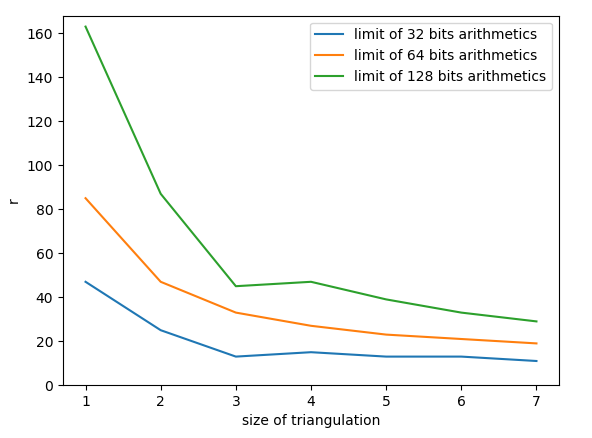}

\caption{Maximal value of the order $r$ for different triangulation sizes such that the error on the multi-precision computation of $\TV_r$ is smaller than $5\%$ on at least half of the census, for 32, 64 and 128 bits.}
\label{fig:multi}

\end{figure}

The weights for Turaev-Viro invariants~\cite{turaev92-invariants} are complex numbers, and the Turaev-Viro invariant $\TV_r(\tri)$ of an $n$-tetrahedra triangulation is made of an exponential large ($r^{O(n)}$) sum of products of $O(n)$ weights.
As a consequence, arithmetic errors may accumulate. We implement multi-precision arithmetic using the {\tt MPFR} library~\cite{MPFR}, within the backtracking algorithm of {\tt Regina}~\cite{regina}. 
Figure 6 illustrates the importance of high-precision arithmetics for computing Turaev-Viro invariants. For a given arithmetic precision, the curve in Figure 6 illustrates the maximal value of $r$ for which the output of the backtrack algorithm is \emph{correct} on more than half of the census. We define \emph{correctness} by a relative error of less than $5\%$. We note that for single precision $64$-bits arithmetic, the computation is already incorrect for values of $r \geq 33$ on more than half the triangulations larger than $3$ tetrahedra.

Note that the estimator of the number of admissible colorings from Section~\ref{sec:ehrhart} can be used to estimate the number of arithmetic operations to be performed by the backtracking algorithm, and in turn set the necessary precision of arithmetic numbers to compute a provably correct approximation of $\TV_r$. However, bounds are pessimistic, and we use the following practical approach.

For each order $r$, we compute each invariant at least twice: with a given number of bits and with twice that number. If the relative difference between the two quantities is larger than a given threshold, we start again with twice the number of bits for the two values until the gap is small enough. For the following order, keep the numbers of bits that were sufficient at the previous order $r$. We begin the computation of the invariants of order $r=3$ using $128$ and $256$ bits. 

Note that the Turaev-Viro invariant can be equal to zero, in that case the relative difference between our two computations may be large. We used a threshold bellow which the invariant is considered equal to zero. 

We validated this approach by testing it for some small triangulations and for small orders on larger triangulations and comparing the results to the same computation with a a larger number of bits. The use of multi-precision produces a large impact on computation times. Figure~\ref{fig:pathos} shows the difference between the backtracking algorithm with single precision arithmetic (leading to erroneous values of $\TV_r$ already for small $r$), the backtracking algorithm using exact encoding of algebraic numbers, and the multi-precision arithmetic we implemented. Both exact and multi-precision arithmetic lead to correct computation of $\TV_r$. We however observe that the multi-precision approach is about $4$ orders of magnitude faster in practice.

Note that the multi-precision curve presents two jumps, one for $r=39$ and one for $r=71$. They correspond to the doubling of arithmetic precision following the strategy described above.

\section{Asymptotics of the sequence of Turaev-Viro invariants}
\label{sec:convergence}

In this section, we verify experimentally and extend, based on our observations, Chen and Yang's volume conjecture~\cite{ChenY2018}, and its extension by Detcherry, Kalfagianni, and Yang~\cite{DetcherryKY2018}, for closed 3-manifolds. The conjecture asserts that the growth rate of the Turaev-Viro invariants of a manifold $\M$ is governed by the {\em simplicial volume} of $\M$. Note that there are several other volume conjectures for manifolds (non-compact, with boundary, etc), but we state only the case of closed manifolds.

\subsection{Volume conjecture for 3-manifolds}

The \emph{simplicial volume} of a (compact oriented) manifold $M$, denoted by $||M||$, is a topological invariant introduced by Gromov~\cite{zbMATH03816552} as a norm on singular homology. Two very important families of manifolds appear in our experiments, \emph{graph manifolds} and \emph{hyperbolic} manifolds.

All closed orientable irreducible 3-manifolds that admit triangulations with at most $8$ tetrahedra are \emph{graph manifolds}~\cite{Matveev07AlgorTopClassif3Mflds}, and they all have simplicial volume $0$. 

There exist closed hyperbolic 3-manifolds, \ie, manifolds admitting a (unique) complete hyperbolic geometry, that can be triangulated with $9$ tetrahedra. For an hyperbolic 3-manifolds $M$, the simplicial volume $||M||$ is strictly positive, and is equal to the hyperbolic volume of $M$ times a universal constant $1/v_3$. In that sense, the simplicial volume is a generalization of the hyperbolic volume to all compact 3-manifolds. 

By work of Jaco-Shalen-Johannson, and the Thurston-Hamilton-Perelman geometrization theorem (see~\cite[Section~2.4]{Matveev07AlgorTopClassif3Mflds}), graph manifolds and hyperbolic manifolds form the building blocks of any 3-manifold.

A motivation of this paper is the following extension of the following conjecture by Chen, Yang~\cite{ChenY2018} (for hyperbolic manifolds), and Detcherry, Kalfagianni, Yang~\cite{DetcherryKY2018} (for general manifolds):

\begin{conjecture}[C-Y~\cite{ChenY2018}, D-K-Y~\cite{DetcherryKY2018}]
\label{conj:chentv}
For a closed oriented hyperbolic 3-manifold $\M$, let $\TV_r(\M)$, $r \geq 3$, be its Turaev-Viro invariants (at $\SL_2(\C)$ with $q=2$) and let $||\M||$ be its simplicial volume. Then, for $r$ running over all odd integers, 
\begin{equation}
\label{eq:volconjhyp}
\lim_{r\rightarrow + \infty} \frac{2\pi}{r}\log(\TV_r(\M)) = v_3||\M||. 
\end{equation}

\medskip

Additionally, for any closed oriented (not necessarily hyperbolic) 3-manifold $\M$, and $r$ running over all odd integers, 
\begin{equation}
\label{eq:volconjlimsup}
\limsup_{r\rightarrow + \infty} \frac{2\pi}{r}\log(\TV_r(\M)) = v_3||\M||.
\end{equation}
\end{conjecture}

This volume conjecture has attracted attention from the topology community.
Notably, Equation~(\ref{eq:volconjhyp}) has been proved for closed hyperbolic manifolds obtained by integral Dehn surgery along the figure eight knot~\cite{def7c30f33224ebea35e3ddaab3a1b1b}. Additionally, any closed 3-manifold $\M$ with $||\M||=0$ satisfies Equation~(\ref{eq:volconjlimsup})~\cite{detcherry2017gromov}. Further, a convergence estimate has been proved for {\em Seifert fibered spaces} (a subclass of graph manifolds), where $\frac{2\pi}{r}\log(\TV_r(\M)) \in O(\log r / r)$.

In the following, we apply the optimized computation of Turaev-Viro invariants described in earlier sections in order to verify experimentally the behavior of the sequences $(\frac{2\pi}{r}\log(TV_r(\M)))_r$ on the census of closed 3-manifolds for large $r$. Our experimental study below allows us to refine the volume conjecture:

\begin{conjecture}
\label{conj:new}
Let $\M$ be any closed oriented (not necessarily hyperbolic) 3-manifold, $\TV_r(\M)$ and $||\M||$ as above. Then, the set of odd integers r such that $\TV_r(\M) \neq 0$ (denoted $R^*_{\M}$) is infinite. Additionally, when $r$ runs over $R^*_{\M}$, we have,
\[\liminf_{r\rightarrow + \infty} \frac{2\pi}{r}\log(TV_r(\M)) = v_3||\M||.\]
\end{conjecture}

This last conjecture completes the above volume conjectures, and both of them fully characterize the asymptotics behavior of the Turaev-Viro sequence of invariants, for any closed $3$-manifolds.

Additionally, for a fixed 3-manifold $M$, we observe, on all manifolds studied, a rate of convergence of the sequence $\frac{2\pi}{r}\log(TV_r(\M))$ towards its limit $v_3||\M||$, of order $\widetilde{O}(\frac{1}{r})$, where $\widetilde{O}$ hides a poly-log factor in $r$.

\paragraph{Measuring the speed of convergence.} Let $\M$ be a manifold for which all Turaev-Viro invariants have been computed until some index $r_{\max}$. Consider the quantity:
\[ 
	S_r(\M) =  \max_{r\leq k\leq r_{\max}}\left(\left|\ \frac{2\pi}{k}\log(TV_k(\M)) - v_3||M||\ \right| \right).
\]
This quantity decreases when $r$ grows, and allows us to study the convergence of the sequence $\frac{2\pi}{r}\log(TV_r(\M))$. We study two sets of manifolds.

\begin{figure}[t]
\includegraphics[width=0.7\columnwidth]{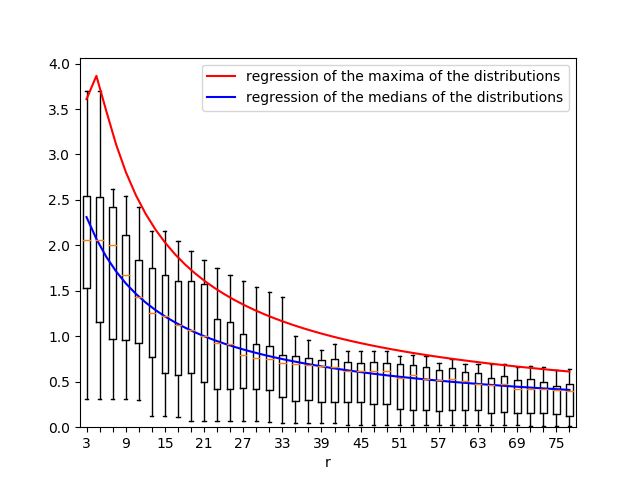}
\caption{Distribution of $(S_r(\M))_r$ for all manifolds with minimal triangulation of size $5$ in the census. The maxima and medians of the distributions are interpolated with fitting Model~\ref{model:1}.}
\label{fig:Sall5}
\end{figure}

\begin{figure}[t]
\includegraphics[width=0.7\columnwidth]{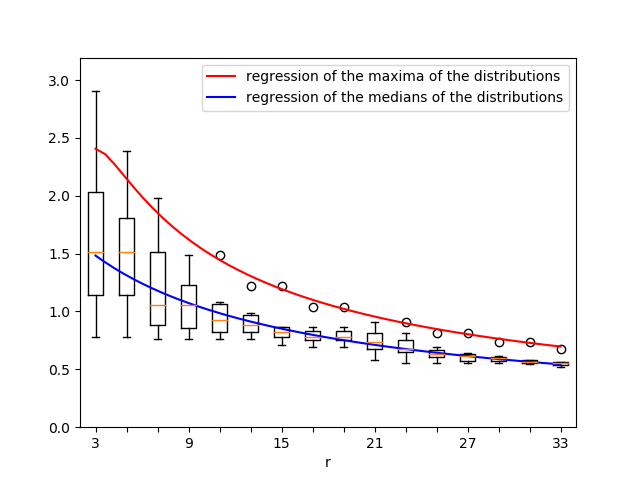}
\caption{Graph of $(S_r(\M))_r$ for non-Seifert fibered graph manifolds with $7$-tetrahedra triangulations. They are obtained by gluing two Seifert fibered spaces along their boundaries. The regression with Model~\ref{model:2} is shown in red.}
\label{fig:graphmanconv}
\end{figure}

\subsection{Graph manifolds with up to $7$ tetrahedra}

In this section we are considering all manifolds with minimal triangulations with at most $6$ tetrahedra, and some triangulations with $7$ tetrahedra. As per the discussion above, they all have simplicial volume $0$. Several manifolds from the census have recurring values of $r$ for which their Turaev-Viro invariant equals $0$ (and the term $\log \TV_r$ is undefined). For all manifolds $\M$, we only consider orders $r \in R^*_\M$ for which $\TV_r(\M) \neq 0$ in the experiments. No sequence appear to become uniformly trivial.

\medskip

Figure~\ref{fig:Sall5} presents, for every closed $3$-manifolds whose minimal triangulation admits $5$ tetrahedra, the distribution of the quantity $S_r(\M)$ for an $r_{\max} = 77$. Appendix~\ref{app:figTV} (Figures~\ref{fig:all1},~\ref{fig:all2},~\ref{fig:all3},~\ref{fig:all4},~\ref{fig:all6}) introduces similar plots for the other minimal triangulations with $\leq 6$ tetrahedra. All sequences appear to converge towards $0$. We interpolate the maximal and median values of $S_r$ over the census with a fitting model:
\begin{equation}
\tag{1}
\label{model:1}
x\mapsto \frac{a\log(x+b)}{x+b}. 
\end{equation}

All manifolds in the census with less than $6$ tetrahedra are Seifert fibered spaces, and the $O(\log r / r)$ behavior is proved for the $\limsup$ of the sequence~\cite{detcherry2017gromov} for such spaces. Figure~\ref{fig:Sall5} emphasizes that the $\liminf$ of the sequence, for all values for which $\TV_r \neq 0$, also converges towards the quantity $v_3||\M||=0$, with experimental convergence $O(\log r / r)$. 

We draw similar conclusions when extending our analysis to the first seven $7$-tetrahedra triangulations of graph manifolds (volume simplicial $0$), that are not Seifert fibered ; see Figure~\ref{fig:graphmanconv}. 

These experiments motivate the new conjecture~\ref{conj:new}, describing fully the convergence of the sequence of Turaev-Viro invariants for all graph manifolds.

\medskip

Additionally, when comparing similar plots for minimal triangulations of $\leq 6$ tetrahedra (see Appendix~\ref{app:figTV}), we notice that the fitting curves, describing the convergence of the sequence, are almost independent of the number of tetrahedra of the triangulations. This suggests the existence of a \emph{slowly growing function} $f$, in the minimal number of tetrahedra, such that for any graph manifold $\M$ that can be triangulated with $n$ tetrahedra, if $\TV_r(\M) \neq 0$ then: 
\[	
   \left| \frac{2 \pi}{r}\log(\TV_r(\M)) \right| \leq f(n) \cdot \frac{\log r}{r}
\]

\subsection{Hyperbolic manifolds with $9$ tetrahedra}
We verify experimentally Conjecture~\ref{conj:chentv} for the four $9$-tetrahedra closed hyperbolic $3$-manifolds. For $9$-tetrahedra triangulations, the computation are much more challenging, and we compute the Turaev-Viro invariants up to $r_{\max} = 39$.
Figure~\ref{fig:seqHyp1} presents the sequence $2\pi/r \log \TV_r(\M)$ of Conjecture~\ref{conj:chentv} for the Weeks manifold, and we observe its convergence towards $v_3||\M|| = 0.9427\ldots$ (the {\em hyperbolic volume} of the manifold). The sequence is interpolated  with the following fitting model:
\begin{equation}
\tag{2}
\label{model:2}
x\mapsto \frac{a}{x+b} + c.
\end{equation}

We observe a similar converging behavior for the other three hyperbolic manifolds with a nine tetrahedra triangulation (see Figures~\ref{fig:b098},~\ref{fig:b101} and~\ref{fig:b126} in Appendix~\ref{app:figTV}).

We assume that we could not observe the potential $\log x$ factors, as in Model~\ref{model:1}, due to lack of data points. However, the $9$-tetrahedra closed hyperbolic $3$-manifold show the convergence behavior predicted by the volume conjecture, and the convergence of the sequence is experimentally in $\widetilde{O}{\left(\frac{1}{r}\right)}$, where $\widetilde{O}$ hide poly-log factors in $r$.

\begin{figure}[t]
\includegraphics[width=0.7\columnwidth]{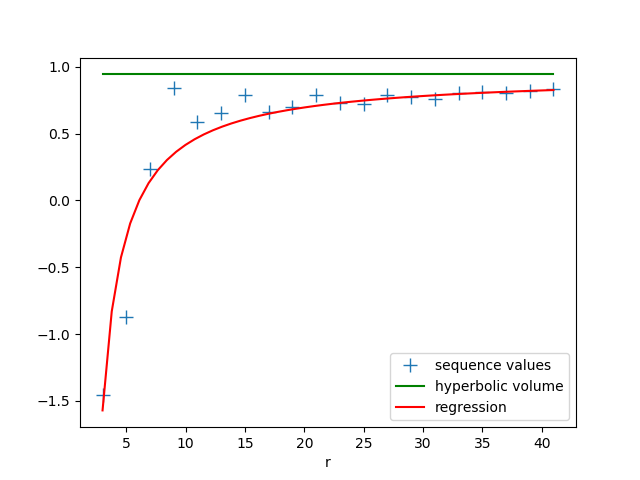}
\caption{Graph of the sequence of Conjecture~\ref{conj:chentv} for the Weeks manifold. The expected limit of the sequence is shown in green. The regression with Model~\ref{model:2} is shown in red.}
\label{fig:seqHyp1}
\end{figure}

\section*{Conclusion}
\label{sec:conclusion}

The Turaev-Viro invariants constitutes an important family of topological invariants for three manifolds. Their combinatorial nature led to several algorithms for their computations and they are efficient to distinguish different manifolds. The recent volume conjecture of Chen and Yang stressed the importance of studying $(\TV_r)_{r \geq 3}$ for large $r$, which is hard in practice.

In this paper, we aimed at computing $(\TV_r)$ for large $r$. To that extend, we looked at two algorithms: a FPT algorithm, theoretically faster but memory consuming, and a backtracking one, with negligible memory cost and good performances on small triangulations.

The contributions of this paper are twofold. First, we establish a easily computable estimator of the complexity
of a triangulation with respect to the backtracking algorithm. This is based on the estimation of the number of integer points inside a polytope and done using Ehrhart polynomial. It leads to a preprocessing that minimize the size of the search tree for the backtracking algorithm. Enriched with the use of multi-precision arithmetics, this approach has allowed the computation of sequences of Turaev-Viro invariants for manifolds with small triangulations.
In turn, these experiments have made possible the elaboration of new conjectures extending the seminal one of Chen and Yang.

While we pursue our computations with increasingly large triangulations, the asymptotic behaviors becomes harder to observe. Future work in this direction will require a finer understanding of the complexity of computing Turaev-Viro invariants, in particular in relation of topological properties of the manifold studied ; see for example~\cite{DBLP:journals/focm/MariaS20}.

\newpage

\appendix

\begin{figure*}
\section{Additional data about the Turaev-Viro sequence}
\label{app:figTV}
This section provides additional experiments for Section~\ref{sec:convergence}.

     \centering
     \begin{subfigure}[b]{0.49\textwidth}
         \centering
         \includegraphics[width=\columnwidth]{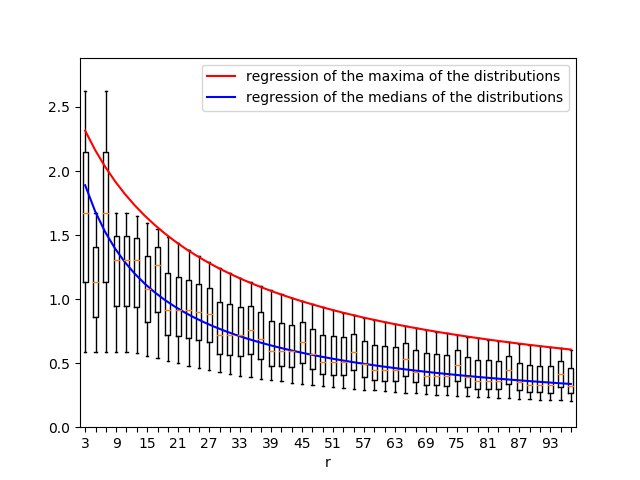}
         \caption{$n=1$}
         \label{fig:all1}
     \end{subfigure}
     \hfill
     \begin{subfigure}[b]{0.49\textwidth}
         \centering
         \includegraphics[width=\columnwidth]{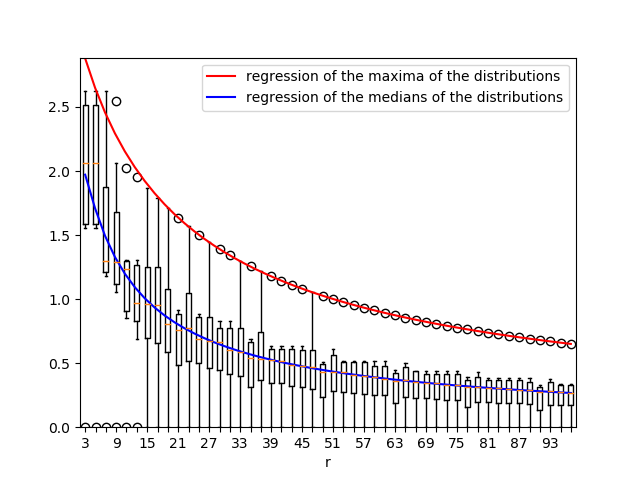}
         \caption{$n=2$}
         \label{fig:all2}
     \end{subfigure}
     \hfill
     \begin{subfigure}[b]{0.49\textwidth}
         \centering
         \includegraphics[width=\columnwidth]{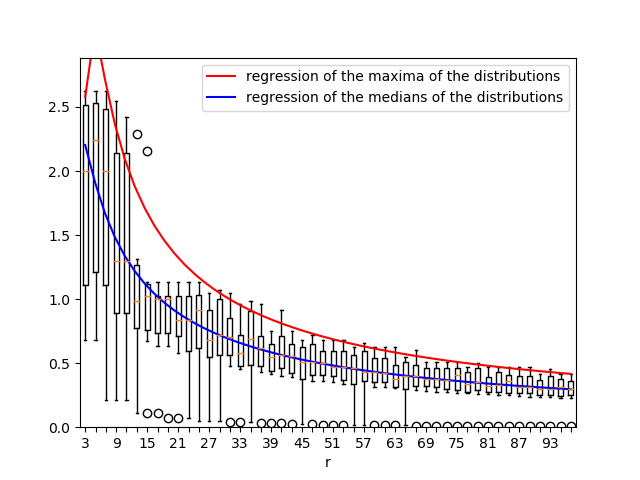}
         \caption{$n=3$}
         \label{fig:all3}
     \end{subfigure}
     \hfill
     \begin{subfigure}[b]{0.49\textwidth}
         \centering
         \includegraphics[width=\columnwidth]{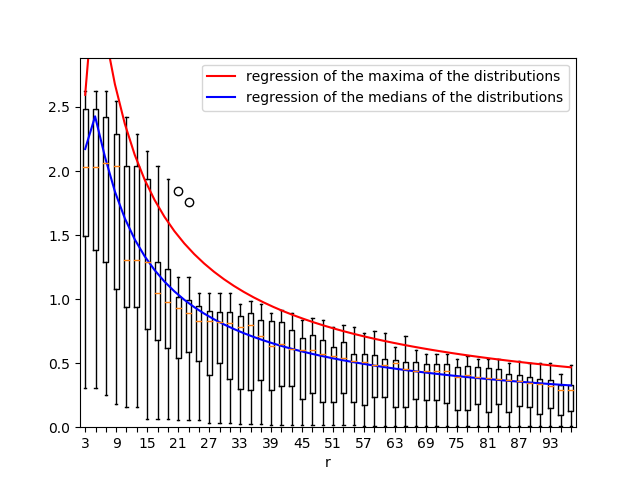}
         \caption{$n=4$}
         \label{fig:all4}
     \end{subfigure}
     \hfill
     \begin{subfigure}[b]{0.49\textwidth}
         \centering
         \includegraphics[width=\columnwidth]{all_5.png}
         \caption{$n=5$}
         \label{fig:all5}
     \end{subfigure}
     \hfill
     \begin{subfigure}[b]{0.49\textwidth}
         \centering
         \includegraphics[width=\columnwidth]{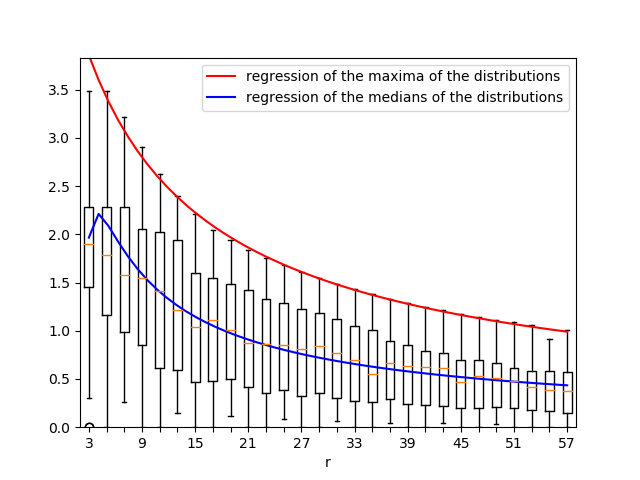}
         \caption{$n=6$}
         \label{fig:all6}
     \end{subfigure}
        \caption{Distribution of $(S_r(\M))_r$ for all manifolds with minimal triangulation of size from one (a) to six (f). Are shown the regressions of the maxima and the medians of the distributions according to Model~\ref{model:1}.}
\end{figure*}

\begin{figure*}[ht]
     \centering
     \begin{subfigure}[b]{0.49\textwidth}
         \centering
         \includegraphics[width=\columnwidth]{hyp_cv_9_0_942.png}
         \caption{$0.9427\ldots$}
         \label{fig:Sr094}
     \end{subfigure}
     \hfill
     \begin{subfigure}[b]{0.49\textwidth}
         \centering
         \includegraphics[width=\columnwidth]{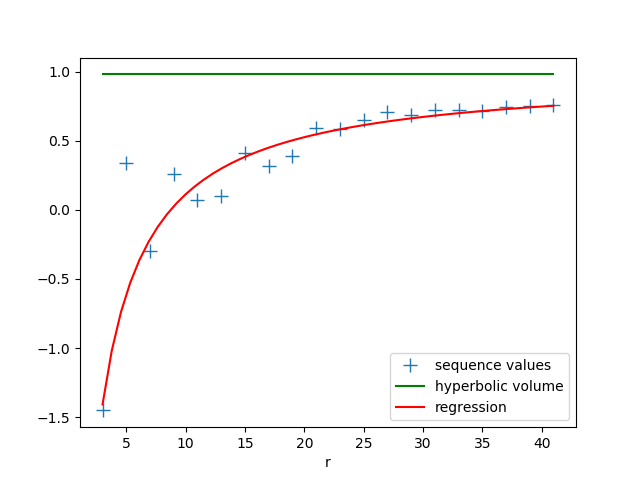}
         \caption{$0.9814\ldots$}
         \label{fig:b098}
     \end{subfigure}
     \hfill
     \begin{subfigure}[b]{0.49\textwidth}
         \centering
         \includegraphics[width=\columnwidth]{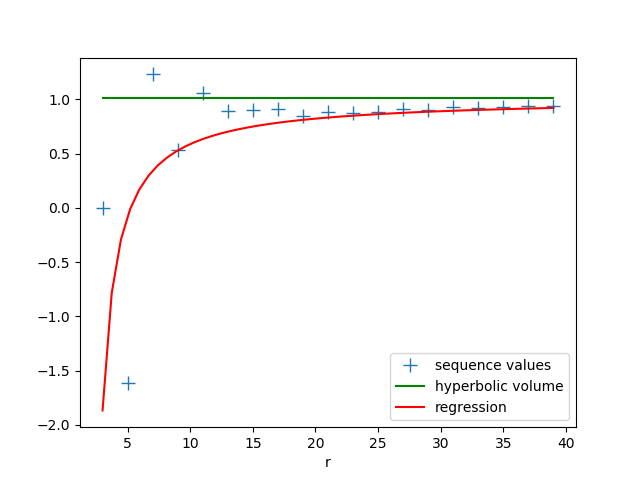}
         \caption{$1.0149\ldots$}
         \label{fig:b101}
     \end{subfigure}
     \hfill
     \begin{subfigure}[b]{0.49\textwidth}
         \centering
         \includegraphics[width=\columnwidth]{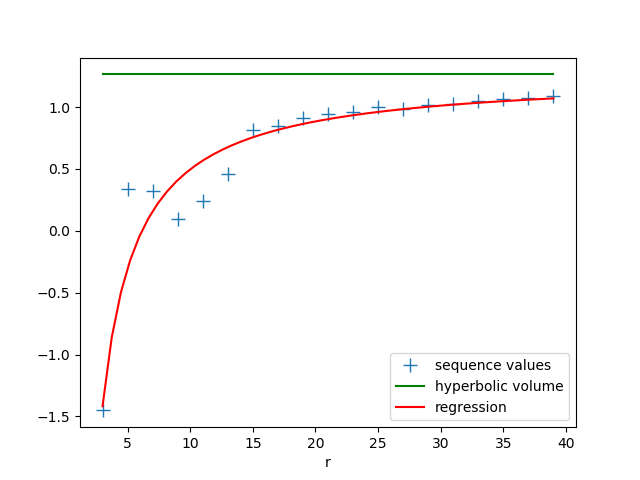}
         \caption{$1.2637\ldots$}
         \label{fig:b126}
     \end{subfigure}
        \caption{Graph of the sequences of Conjecture~\ref{conj:chentv} for all closed hyperbolic 3-manifolds admitting a nine tetrahedra triangulation. The regression with Model~\ref{model:2} is shown in red and the hyperbolic volumes are shown in green. The captions correspond to the hyperbolic volumes of the manifolds.}
\end{figure*}

\begin{figure*}[ht]
     \centering
     \begin{subfigure}[b]{0.49\textwidth}
         \centering
         \includegraphics[width=\columnwidth]{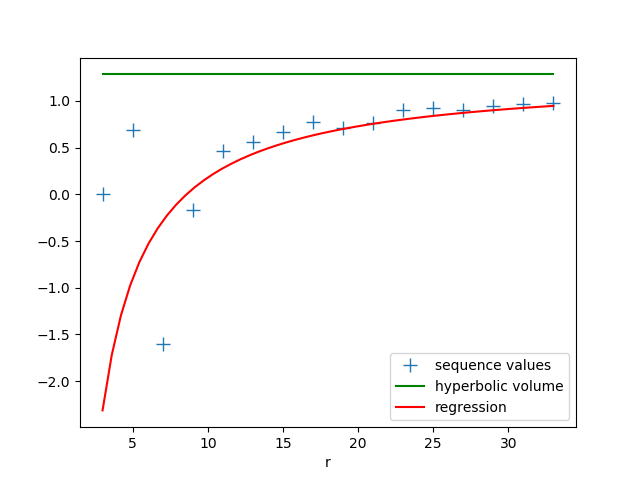}
         \caption{$1.2844\ldots$}
     \end{subfigure}
     \hfill
     \begin{subfigure}[b]{0.49\textwidth}
         \centering
         \includegraphics[width=\columnwidth]{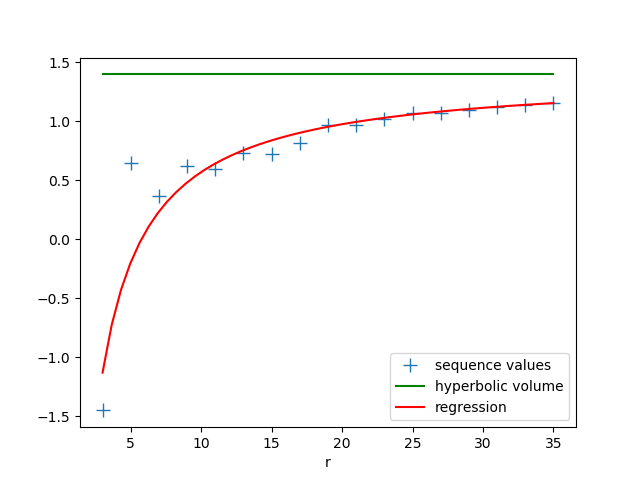}
         \caption{$1.3985\ldots$}
     \end{subfigure}
     \hfill
     \begin{subfigure}[b]{0.49\textwidth}
         \centering
         \includegraphics[width=\columnwidth]{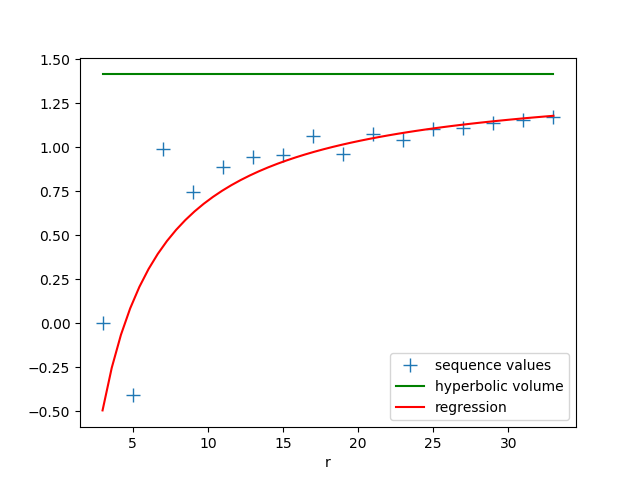}
         \caption{$1.4140\ldots,\ \Z_6$}
     \end{subfigure}
     \hfill
     \begin{subfigure}[b]{0.49\textwidth}
         \centering
         \includegraphics[width=\columnwidth]{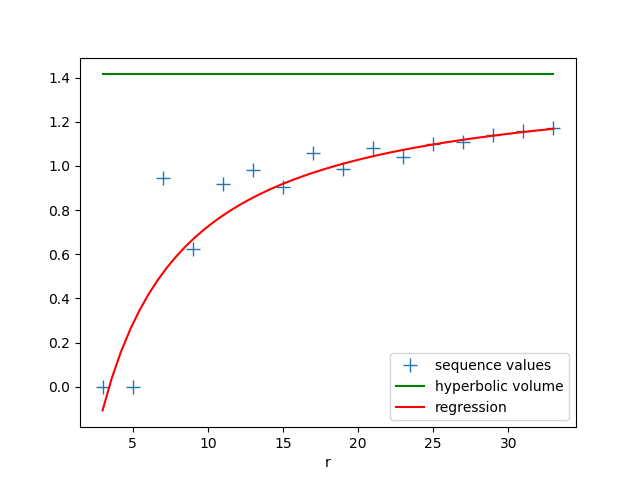}
         \caption{$1.4140\ldots,\ \Z_{10}$}
     \end{subfigure}
     \begin{subfigure}[b]{0.49\textwidth}
         \centering
         \includegraphics[width=\columnwidth]{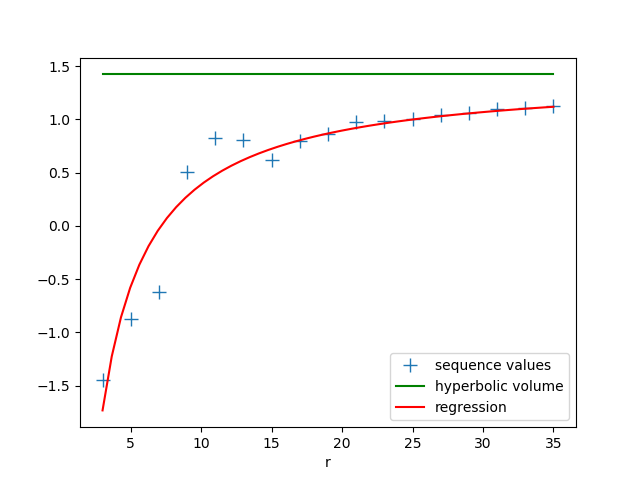}
         \caption{$1.4236\ldots$}
     \end{subfigure}
        \caption{Graph of the sequences of Conjecture~\ref{conj:chentv} for different closed hyperbolic 3-manifolds admitting a ten tetrahedra triangulation. The regression with Model~\ref{model:2} is shown in red and the hyperbolic volumes are shown in green. The captions correspond to the hyperbolic volumes of the manifolds, and occasionally followed by the first homology group, and by the length of the shortest closed geodesic to clarify the ambiguities.}
\end{figure*}

\begin{figure*}[ht]
     \centering
     \begin{subfigure}[b]{0.49\textwidth}
         \centering
         \includegraphics[width=\columnwidth]{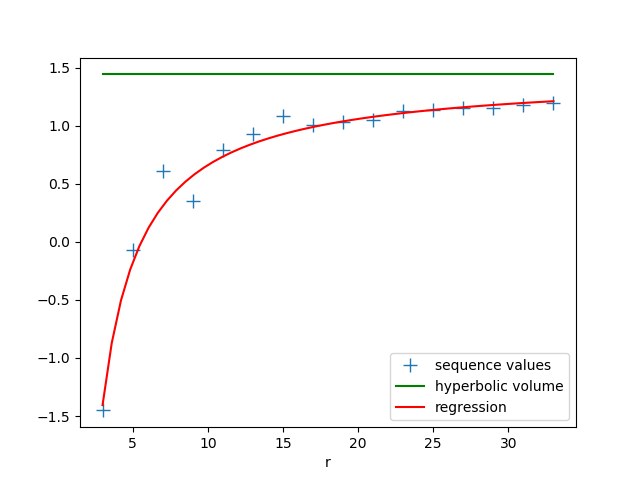}
         \caption{$1.4406\ldots$}
     \end{subfigure}
     \hfill
     \begin{subfigure}[b]{0.49\textwidth}
         \centering
         \includegraphics[width=\columnwidth]{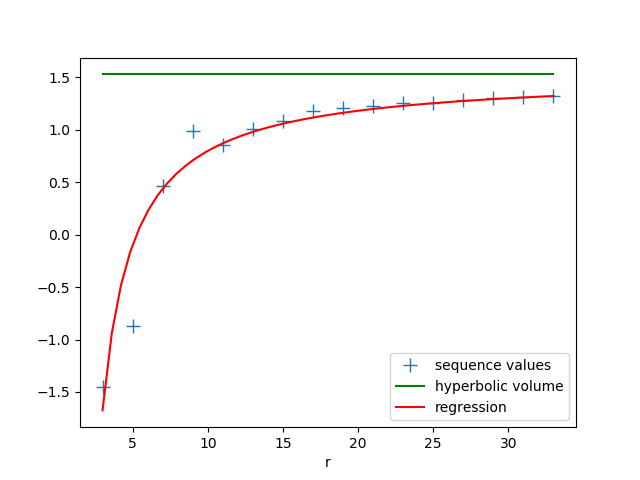}
         \caption{$1.5294\ldots$}
     \end{subfigure}
     \hfill
     \begin{subfigure}[b]{0.49\textwidth}
         \centering
         \includegraphics[width=\columnwidth]{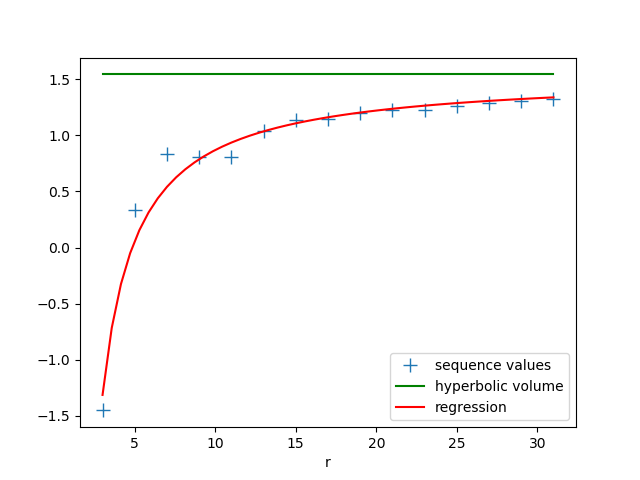}
         \caption{$1.5435\ldots,\ \Z_{35}$}
     \end{subfigure}
     \hfill
     \begin{subfigure}[b]{0.49\textwidth}
         \centering
         \includegraphics[width=\columnwidth]{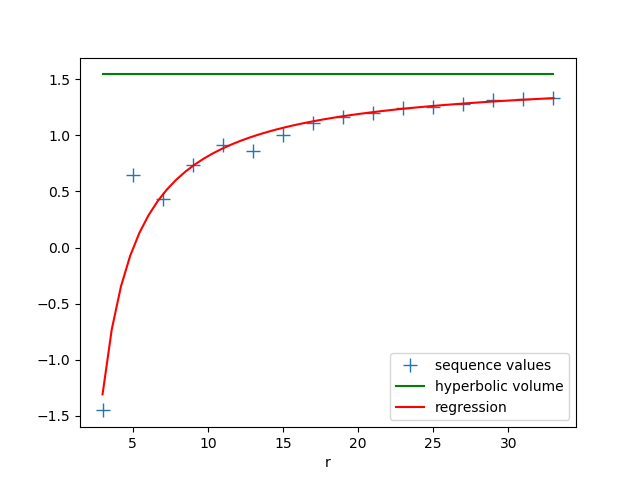}
         \caption{$1.5435\ldots,\ \Z_{21}$}
     \end{subfigure}
     \begin{subfigure}[b]{0.49\textwidth}
         \centering
         \includegraphics[width=\columnwidth]{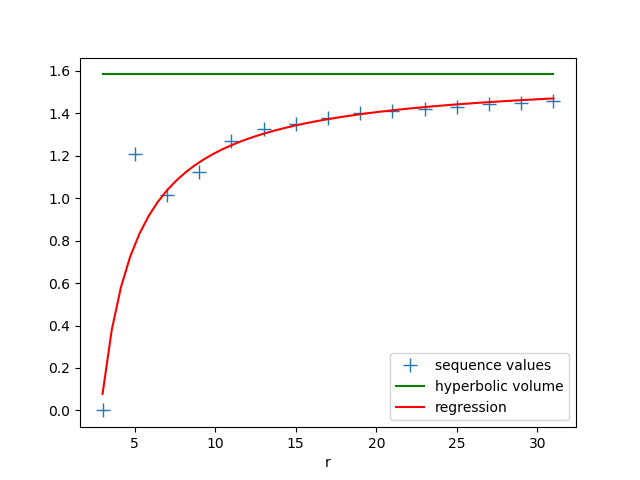}
         \caption{$1.5831\ldots,\ \Z_{40}$}
     \end{subfigure}
        \caption{Graph of the sequences of Conjecture~\ref{conj:chentv} for different closed hyperbolic 3-manifolds admitting a ten tetrahedra triangulation. The regression with Model~\ref{model:2} is shown in red and the hyperbolic volumes are shown in green. The captions correspond to the hyperbolic volumes of the manifolds, and occasionally followed by the first homology group, and by the length of the shortest closed geodesic to clarify the ambiguities.}
\end{figure*}

\begin{figure*}[ht]
     \centering
     \begin{subfigure}[b]{0.49\textwidth}
         \centering
         \includegraphics[width=\columnwidth]{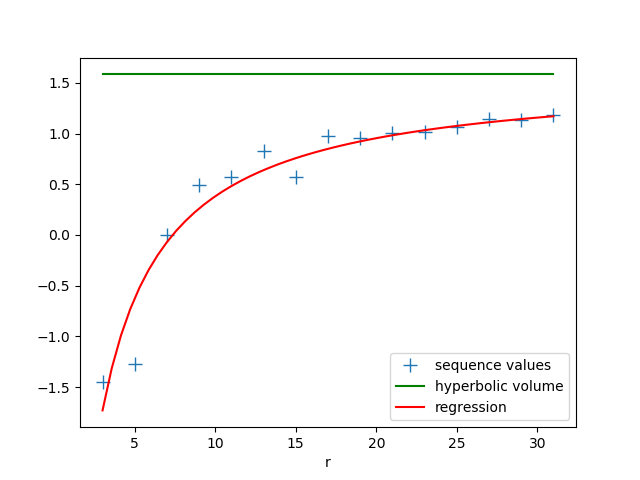}
         \caption{$1.5831\ldots,\ \Z_{21}$}
     \end{subfigure}
     \hfill
     \begin{subfigure}[b]{0.49\textwidth}
         \centering
         \includegraphics[width=\columnwidth]{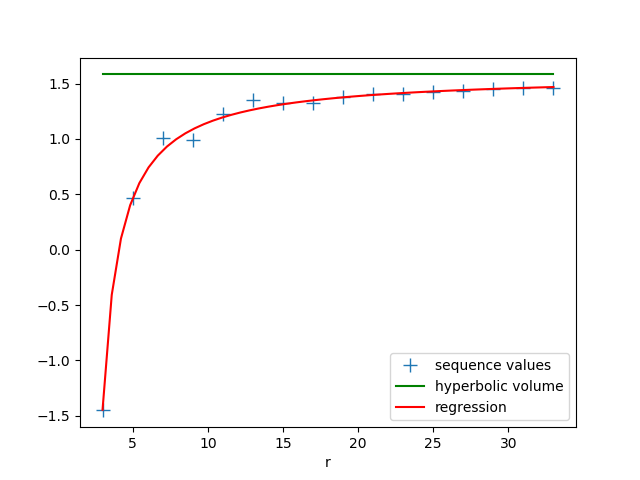}
         \caption{$1.5831\ldots,\ \Z_{3}+\Z_9$}
     \end{subfigure}
     \hfill
     \begin{subfigure}[b]{0.49\textwidth}
         \centering
         \includegraphics[width=\columnwidth]{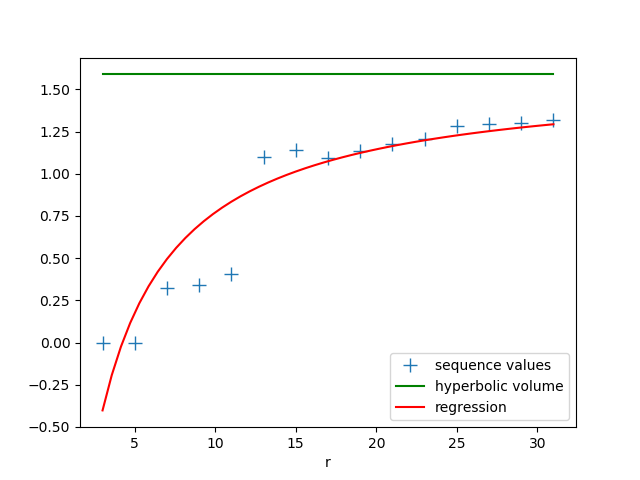}
         \caption{$1.5886\ldots,\ \Z_{30},\ 0.3046$}
     \end{subfigure}
     \hfill
     \begin{subfigure}[b]{0.49\textwidth}
         \centering
         \includegraphics[width=\columnwidth]{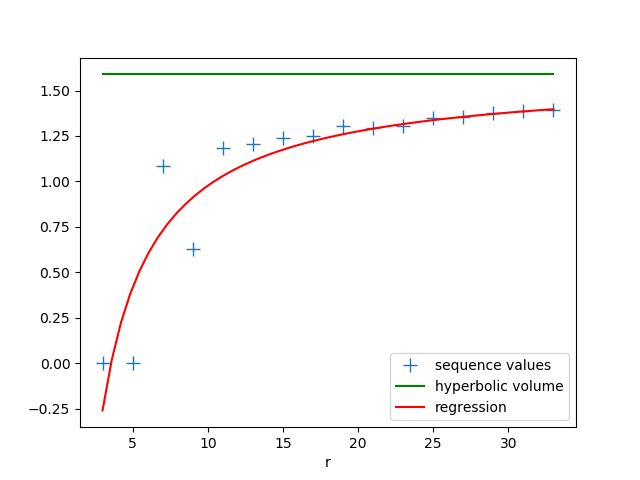}
         \caption{$1.5886\ldots,\ \Z_{30},\ 0.5345$}
     \end{subfigure}
     \begin{subfigure}[b]{0.49\textwidth}
         \centering
         \includegraphics[width=\columnwidth]{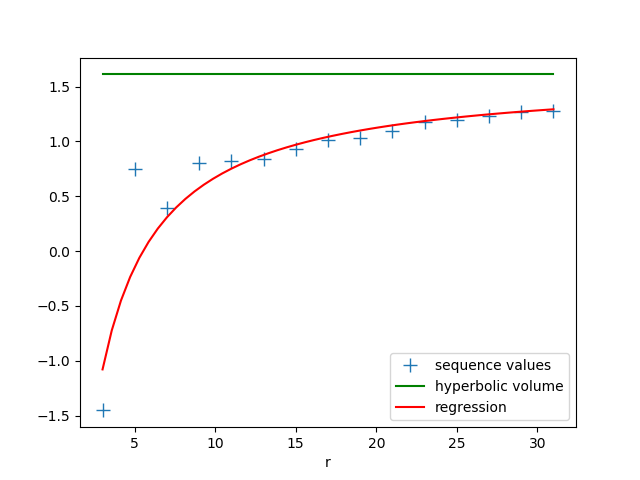}
         \caption{$1.6104\ldots$}
     \end{subfigure}
        \caption{Graph of the sequences of Conjecture~\ref{conj:chentv} for different closed hyperbolic 3-manifolds admitting a ten tetrahedra triangulation. The regression with Model~\ref{model:2} is shown in red and the hyperbolic volumes are shown in green. The captions correspond to the hyperbolic volumes of the manifolds, and occasionally followed by the first homology group, and by the length of the shortest closed geodesic to clarify the ambiguities.}
\end{figure*}

\begin{figure*}[ht]
     \centering
     \begin{subfigure}[b]{0.49\textwidth}
         \centering
         \includegraphics[width=\columnwidth]{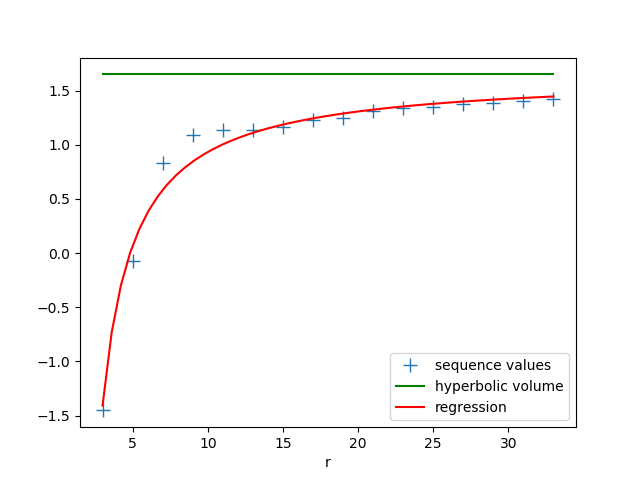}
         \caption{$1.6496\ldots,\ \Z_7$}
     \end{subfigure}
     \hfill
     \begin{subfigure}[b]{0.49\textwidth}
         \centering
         \includegraphics[width=\columnwidth]{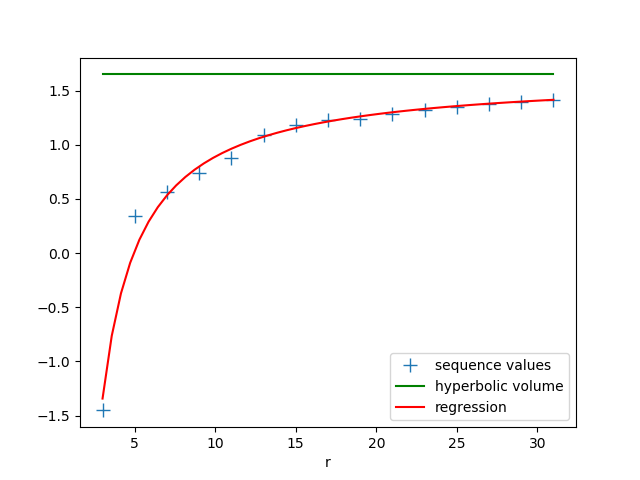}
         \caption{$1.6496\ldots,\ \Z_{15}$}
     \end{subfigure}
     \hfill
     \begin{subfigure}[b]{0.49\textwidth}
         \centering
         \includegraphics[width=\columnwidth]{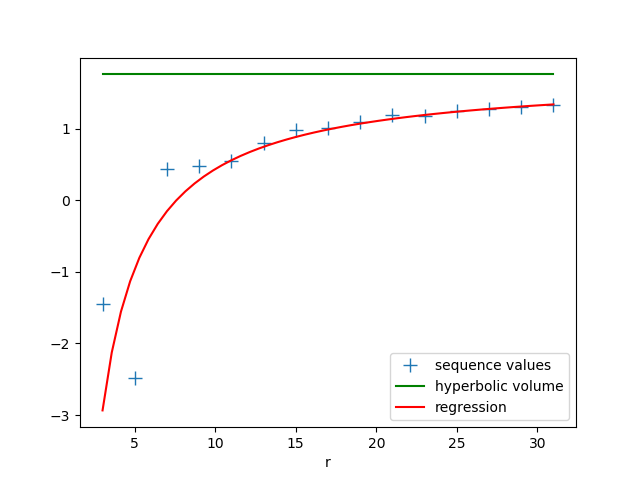}
         \caption{$1.7571\ldots$}
     \end{subfigure}
     \hfill
     \begin{subfigure}[b]{0.49\textwidth}
         \centering
         \includegraphics[width=\columnwidth]{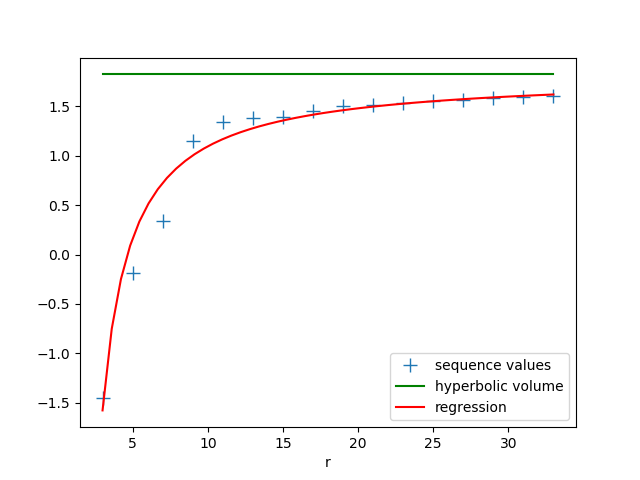}
         \caption{$1.8243\ldots$}
     \end{subfigure}
     \begin{subfigure}[b]{0.49\textwidth}
         \centering
         \includegraphics[width=\columnwidth]{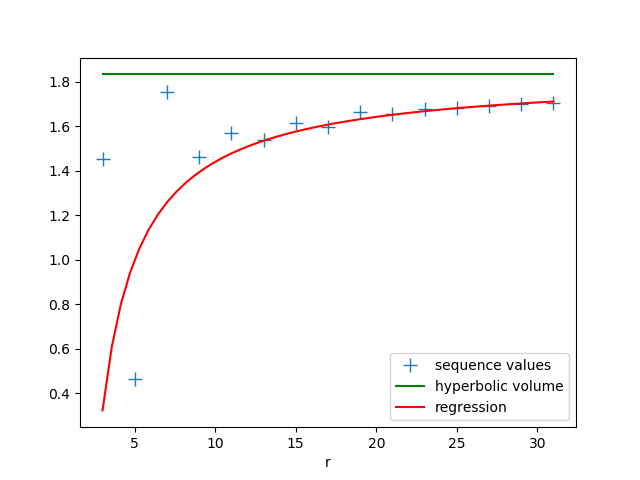}
         \caption{$1.8319\ldots$}
     \end{subfigure}
        \caption{Graph of the sequences of Conjecture~\ref{conj:chentv} for different closed hyperbolic 3-manifolds admitting a ten tetrahedra triangulation. The regression with Model~\ref{model:2} is shown in red and the hyperbolic volumes are shown in green. The captions correspond to the hyperbolic volumes of the manifolds, and occasionally followed by the first homology group, and by the length of the shortest closed geodesic to clarify the ambiguities.}
\end{figure*}

\begin{figure*}[ht]
     \centering
     \begin{subfigure}[b]{0.49\textwidth}
         \centering
         \includegraphics[width=\columnwidth]{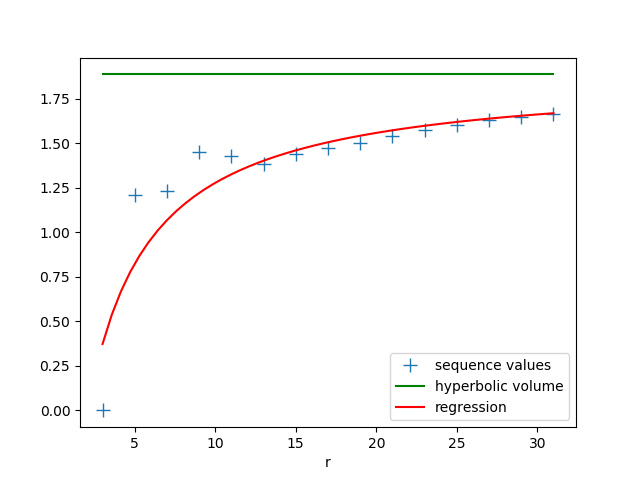}
         \caption{$1.8854\ldots,\ \Z_{40}$}
     \end{subfigure}
     \hfill
     \begin{subfigure}[b]{0.49\textwidth}
         \centering
         \includegraphics[width=\columnwidth]{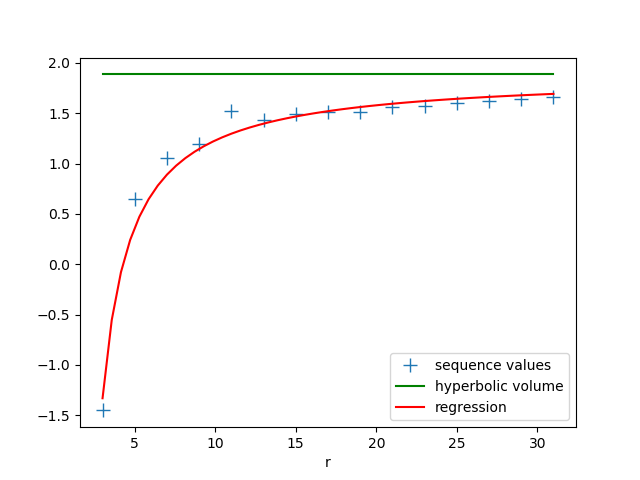}
         \caption{$1.8854\ldots,\ \Z_{7}+\Z_{7}$}
     \end{subfigure}
     \hfill
     \begin{subfigure}[b]{0.49\textwidth}
         \centering
         \includegraphics[width=\columnwidth]{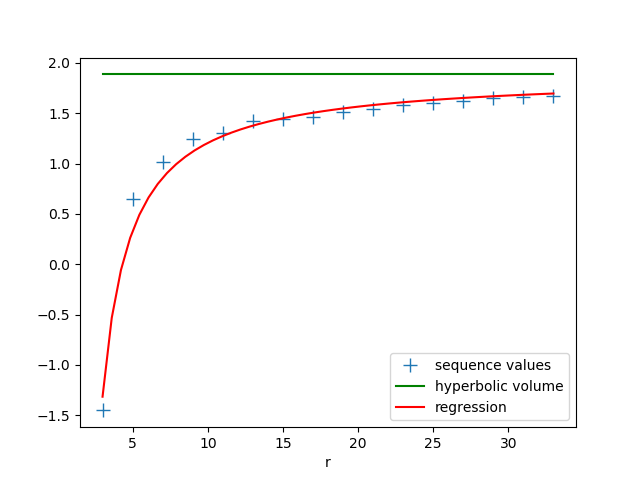}
         \caption{$1.8854\ldots,\ \Z_{39}$}
     \end{subfigure}
     \hfill
     \begin{subfigure}[b]{0.49\textwidth}
         \centering
         \includegraphics[width=\columnwidth]{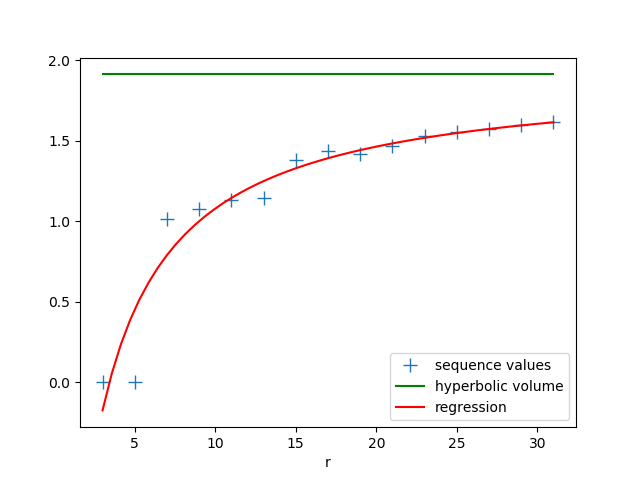}
         \caption{$1.9108\ldots$}
     \end{subfigure}
     \begin{subfigure}[b]{0.49\textwidth}
         \centering
         \includegraphics[width=\columnwidth]{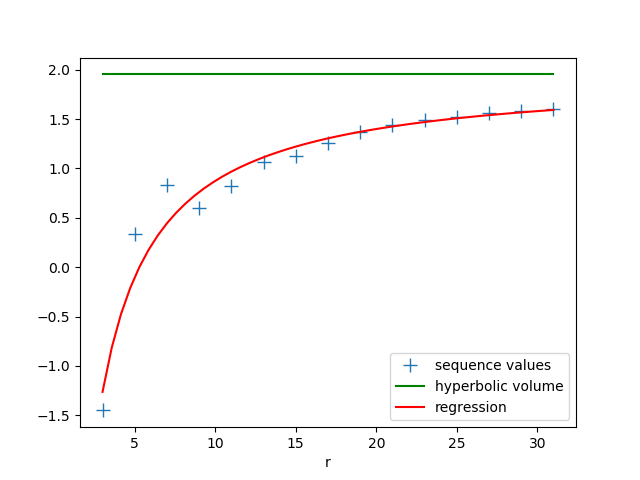}
         \caption{$1.9537\ldots$}
     \end{subfigure}
        \caption{Graph of the sequences of Conjecture~\ref{conj:chentv} for different closed hyperbolic 3-manifolds admitting a ten tetrahedra triangulation. The regression with Model~\ref{model:2} is shown in red and the hyperbolic volumes are shown in green. The captions correspond to the hyperbolic volumes of the manifolds, and occasionally followed by the first homology group, and by the length of the shortest closed geodesic to clarify the ambiguities.}
\end{figure*}

\end{document}